\documentclass[a4paper,UKenglish,autoref,thm-restate,cleveref]{lipics-v2021}
\pdfoutput=1

\usepackage{mathtools}
\usepackage{upgreek}
\usepackage{tikz}

\hideLIPIcs
\nolinenumbers
\newboolean{arxiv}
\setboolean{arxiv}{true}

\newcommand{\Prop}{\mathsf{hProp}}
\newcommand{\Set}{\mathsf{hSet}}
\newcommand{\UU}{\mathcal{U}}
\newcommand{\N}{\mathbb N}
\newcommand{\Bool}{\mathbf{2}}
\newcommand{\Unit}{\mathbf{1}}
\newcommand{\Empty}{\mathbf{0}}
\newcommand{\bff}{\mathsf{ff}}
\newcommand{\btt}{\mathsf{tt}}

\newcommand{\dissum}{\mathop{\uplus}}
\newcommand{\inl}{\mathsf{inl}}
\newcommand{\inr}{\mathsf{inr}}

\newcommand{\fstproj}{\mathsf{fst}}
\newcommand{\sndproj}{\mathsf{snd}}

\newcommand{\blank}{\_}

\renewcommand{\iff}{\leftrightarrow}

\newcommand{\defeq}{\vcentcolon\equiv}

\newcommand{\acc}{\mathsf{acc}}

\newcommand{\cnf}{\mathsf{Cnf}}
\newcommand{\brouwer}{\mathsf{Brw}}
\newcommand{\bookord}{\mathsf{Ord}}

\newcommand{\bzero}{\mathsf{zero}}
\newcommand{\bsuc}{\mathsf{succ}}
\newcommand{\blimit}{\mathsf{limit}}
\newcommand{\bbisim}{\mathsf{bisim}}
\newcommand{\lzero}{\mathord\leq\mbox{-}\mathsf{zero}}
\newcommand{\ltrans}{\mathord\leq\mbox{-}\mathsf{trans}}
\newcommand{\lsuccmono}{\mathord\leq\mbox{-}\mathsf{succ}\mbox{-}\mathsf{mono}}
\newcommand{\lcocone}{\mathord\leq\mbox{-}\mathsf{cocone}}
\newcommand{\llimiting}{\mathord\leq\mbox{-}\mathsf{limiting}}

\newcommand{\toincr}[1]{\xrightarrow{#1}}

\newcommand{\isZero}{\mathsf{isZero}}
\newcommand{\Code}{\mathsf{Code}}
\newcommand{\toCode}{\mathsf{toCode}}
\newcommand{\fromCode}{\mathsf{fromCode}}

\newcommand{\osuc}[1]{#1 \dissum \Unit}
\newcommand{\osup}[1]{\lim #1}

\newcommand{\simplesuc}{\mathsf{suc}}

\newcommand{\LEM}{\mathsf{LEM}}

\newcommand{\TT}{\mathcal T}
\newcommand{\fst}{\mathsf{left}}

\newcommand{\fatplus}{\raisebox{0ex}{\tikz\filldraw[black,x=.7pt,y=.7pt] (0,0) -- ++(3,0) -- ++(0,3) -- ++(1,0) -- ++(0,-3) -- ++(3,0) -- ++(0,-1) -- ++(-3,0) -- ++(0,-3) -- ++(-1,0) -- ++(0,3) -- ++(-3,0) -- cycle;}}

\newcommand{\tz}{\mathrm 0 }
\newcommand{\tom}[2]{{\upomega}^{#1} \, \fatplus \, {#2}}

\newcommand{\tone}{1}

\newcommand{\iszero}{\mathsf{is}\mbox{-}\mathsf{zero}}

\newcommand{\isstrongsuc}{\mathsf{is}\mbox{-}\mathsf{str}\mbox{-}\mathsf{suc}}

\newcommand{\islim}{\mathsf{is}\mbox{-}\mathsf{lim}}

\newcommand{\islimit}{\mathsf{is}\mbox{-}\mathsf{limit}}

\newcommand{\issucof}[2]{#1 \; \mathsf{is}\mbox{-}\mathsf{suc}\mbox{-}\mathsf{of} \; #2}
\let\isupsucof\issucof

\newcommand{\isstrongsucof}[2]{#1 \; \mathsf{is}\mbox{-}\mathsf{str}\mbox{-}\mathsf{suc}\mbox{-}\mathsf{of} \; #2}

\newcommand{\issupof}[2]{#1 \; \mathsf{is}\mbox{-}\mathsf{sup}\mbox{-}\mathsf{of} \; #2}
\let\isupsupof\issupof

\newcommand{\islimof}[2]{#1 \; \mathsf{is}\mbox{-}\mathsf{lim}\mbox{-}\mathsf{of} \; #2}

\newcommand{\CtoB}{\mathsf{CtoB}}
\newcommand{\BtoO}{\mathsf{BtoO}}

\newcommand{\formalisedQED}{\faCog}
\newcommand{\formalisedQEDfull}{\faCogs}
\newcommand{\formalisedQEDsymbol}{\qedsymbol{}\formalisedQED}
\newcommand{\formalisedQEDfullsymbol}{\qedsymbol{}\formalisedQEDfull}
\newcommand{\qedhref}[1]{\qed{}\href{#1}{\formalisedQED}}
\newcommand{\qedfullhref}[1]{\qed\href{#1}{\formalisedQEDfull}}

\title{\texorpdfstring{Connecting Constructive Notions of Ordinals \\ in Homotopy Type Theory}{Connecting Constructive Notions of Ordinals in Homotopy Type Theory}}
\titlerunning{Connecting Constructive Notions of Ordinals}

\author{Nicolai Kraus}{University of Nottingham, UK}{nicolai.kraus@nottingham.ac.uk}{https://orcid.org/0000-0002-8729-4077}{The Royal Society, grant reference URF\textbackslash{}R1\textbackslash{}191055.}
\author{Fredrik Nordvall Forsberg}{University of Strathclyde, UK}{fredrik.nordvall-forsberg@strath.ac.uk}{https://orcid.org/0000-0001-6157-9288}{UK National Physical Laboratory Measurement Fellowship project ``Dependent types for trustworthy tools''.}
\author{Chuangjie Xu}{fortiss GmbH, Germany}{xu@fortiss.org}{https://orcid.org/0000-0001-6838-4221}{The Humboldt Foundation and the LMUexcellent program.}
\authorrunning{N.\ Kraus, F.\ Nordvall Forsberg, and C.\ Xu}

\Copyright{Nicolai Kraus, Fredrik Nordvall Forsberg, and Chuangjie Xu}

\EventEditors{Filippo Bonchi and Simon J. Puglisi}
\EventNoEds{2}
\EventLongTitle{46th International Symposium on Mathematical Foundations of Computer Science (MFCS 2021)}
\EventShortTitle{MFCS 2021}
\EventAcronym{MFCS}
\EventYear{2021}
\EventDate{August 23--27, 2021}
\EventLocation{Tallinn, Estonia}
\EventLogo{}
\SeriesVolume{202}
\ArticleNo{42}

\supplement{A formalisation is available:}
\supplementdetails{Agda source code}{https://bitbucket.org/nicolaikraus/constructive-ordinals-in-hott/src}
\supplementdetails{Agda as html}{https://cj-xu.github.io/agda/constructive-ordinals-in-hott/index.html}

\relatedversion{This paper has appeared in the proceedings of the \emph{46th International Symposium on Mathematical Foundations of Computer Science (MFCS 2021)}. The MFCS version does not contain the \hyperref[sec:appendix]{appendix} included here, but is otherwise essentially identical (including theorem numberings).}

\acknowledgements{%
We thank the participants of the conferences \emph{Developments in Computer Science} and \emph{TYPES}, as well as Helmut Schwichtenberg and Thorsten Altenkirch for fruitful discussions on this work.
We are also grateful to the anonymous reviewers, whose remarks helped us improve the paper.
}

\begin{document}

\maketitle

\begin{abstract}
	In classical set theory, there are many equivalent ways to introduce
	ordinals. In a constructive setting, however, the different notions
	split apart, with different advantages and disadvantages for
	each. We consider three different notions of ordinals in homotopy
	type theory, and show how they relate to each other: A notation
	system based on Cantor normal forms, a refined notion of Brouwer
	trees (inductively generated by zero, successor and countable
	limits), and wellfounded extensional orders. For Cantor
	normal forms, most properties are decidable, whereas for wellfounded
	extensional transitive orders, most are undecidable. Formulations
	for Brouwer trees are usually partially decidable. We demonstrate
	that all three notions have properties expected of ordinals: their
	order relations, although defined differently in each case, are all
	extensional and wellfounded, and the usual arithmetic
	operations can be defined in each case.  We connect these notions by
	constructing structure preserving embeddings of Cantor normal forms
	into Brouwer trees, and of these in turn into wellfounded
	extensional orders.
	We have formalised most of our results in cubical Agda.
\end{abstract}

\section{Introduction}

The use of ordinals is a powerful tool when proving
that processes terminate, when justifying induction and
recursion \cite{dershowitz:termination,Floyd:1967}, or in
(meta)mathematics generally.
Unfortunately, the standard definition of ordinals is not very
well-behaved constructively, meaning that additional work is required before this tool can be deployed in constructive mathematics or program verification tools based on constructive type theory such as Agda~\cite{norell07thesis}, Coq~\cite{Coq} or Lean~\cite{lean}.
Constructively, the classical notion of ordinal fragments into a number of
inequivalent definitions, each with pros and cons. For example,
``syntactic'' ordinal notation
systems~\cite{buchholz:notation,schuette:book,takeuti:book} are
popular with proof theorists, as their concrete character typically
mean that equality and the order relation on ordinals are
decidable. However, truly infinitary operations such as taking the
limit of a countable sequence of ordinals are usually not
constructible. We will consider a simple ordinal notation system based
on Cantor normal forms~\cite{NFXG:three:ord}, designed in such a way
that there are no ``junk'' terms not denoting real ordinals.

Another alternative (based on notation systems by Church~\cite{church:1938} and
Kleene~\cite{kleene:notation-systems}), popular in the functional
programming community, is to consider ``Brouwer ordinal trees''
$\mathcal{O}$ inductively generated by zero, successor and a
``supremum'' constructor
\[
  \mathsf{sup} : (\N \to \mathcal{O}) \to \mathcal{O}
\]
which forms a new tree for every countable sequence of
trees~\cite{brouwer:trees,coquand:ord-in-tt,hancock:thesis}. By the
inductive nature of the definition, constructions on trees can be
carried out by giving one case for zero, one for successors, and one
for suprema, just as in the classical theorem of transfinite
induction. However calling the constructor $\mathsf{sup}$ is wishful
thinking; $\mathsf{sup}(s)$ does not faithfully represent the suprema
of the sequence $s$, since we do not have that e.g.\
$\mathsf{sup}(s_0, s_1, s_2, \ldots) = \mathsf{sup}(s_1, s_0, s_2,
\ldots)$ --- each sequence gives rise to a new tree, rather than
identifying trees representing the same suprema. We use the notion of
higher inductive types~\cite{cubicalhits,lumsdaine:hits} from homotopy
type theory~\cite{hott-book} to remedy the situation and make a type
of Brouwer trees which faithfully represents ordinals. Since our
ordinals now can be infinitary, we lose decidability of equality and order
relations, but we retain the possibility of classifying an ordinal as
a zero, a successor or a limit.

One can also consider extensional wellfounded orders, a variation on the classical set-theoretical
axioms more suitable for a constructive
treatment~\cite{taylor:ordinals}, which was transferred to the setting of
homotopy type theory in the HoTT book~\cite[Chapter 10]{hott-book},
and significantly extended by
Escard\'{o}~\cite{escardo:agda-ordinals}. One is then forced to give
up most notions of decidability --- it is not even possible to decide if a
given ordinal is zero, a successor or a limit. However many operations
can still be defined on such ordinals, and properties such as
wellfoundedness can still be proven. This is also the notion of
ordinal most closely related to the traditional notion, and thus the
most obviously ``correct'' notion in a classical setting.

All in all, each of these approaches gives quite a different feel to
the ordinals they represent: Cantor normal forms emphasise syntactic
manipulations, Brouwer trees how every ordinal can be classified as a
zero, successor or limit, and extensional wellfounded orders the set
theoretic properties of ordinals. As a consequence, each notion of
ordinals is typically used in isolation, with no interaction or
opportunities to transfer constructions and ideas from one setting to
another --- e.g., do the arithmetic operations defined on Cantor
normal forms obey the same rules as the arithmetic operations
defined on Brouwer trees? The goal of this paper is to answer such
questions by connecting together the different notions. We
do this firstly by introducing an abstract axiomatic framework of what
we expect of any notion of ordinal, and explore to what extent the
notions above satisfy these axioms, and secondly by constructing
faithful embeddings between the notions, which shows that they all
represent a correct notion of ordinal from the point of view of
classical set theory.

\paragraph*{Contributions}
\begin{itemize}
\item We identify an axiomatic framework for ordinals and ordinal arithmetic
  that we use to compare the situations above in the setting of homotopy type theory. 
\item We define arithmetic operations on Cantor normal
  forms~\cite{NFXG:three:ord} and prove them uniquely
  correct with respect to our abstract axiomatisation. This notion of correctness has not been verified for Cantor normal forms previously, as far as we know.
\item We construct a higher inductive-inductive type of Brouwer trees,
  and prove that their order is both wellfounded and extensional ---
  properties which do not hold simultaneously for previous definitions
  of ordinals based on Brouwer trees. 
  Further, we define arithmetic operations, and show that
  they are uniquely correct.
\item We prove that the ``set-theoretic'' notion of
  ordinals~\cite[Section 10.3]{hott-book} satisfies our axiomatisation
  of addition and multiplication, and give constructive ``taboos'',
  showing that many operations on these ordinals are not possible
  constructively.  
\item We relate and connect these different notions of ordinals by
  constructing order preserving embeddings from more decidable notions
  into less decidable ones. 
\end{itemize}

\paragraph*{Formalisation and Full Proofs}

We have formalised the material on
Cantor normal forms and Brouwer trees
in cubical Agda~\cite{VMA:cubical:agda} at \url{https://cj-xu.github.io/agda/constructive-ordinals-in-hott/}; see also Escard\'{o}'s
formalisation~\cite{escardo:agda-ordinals} of many results on
``set-theoretic'' ordinals in HoTT.
We have marked theorems with formalised and partly formalised proofs using the QED symbols \formalisedQEDfullsymbol{} and \formalisedQEDsymbol{} respectively; they are also clickable links to the corresponding machine-checked statement.
Moreover, pen-and-paper proofs for all our results can be found in the
\ifthenelse{\boolean{arxiv}}{\hyperref[sec:appendix]{appendix}.}{\href{https://arxiv.org/abs/2104.02549}{arXiv version} of the paper.}

Our formalisation uses the \texttt{\{-\# TERMINATING \#-\}} pragma to
work around one known bug
(\href{https://github.com/agda/agda/issues/4725}{issue \#4725}) and
one limitation of the termination checker of Agda: recursive calls
hidden under a propositional truncation are not seen to be
structurally smaller.  Such recursive calls when proving a proposition
are justified by the eliminator presentation of \cite{gabe:phd}
(although it would be non-trivial to reduce our mutual definitions to
eliminators).

\section{Underlying Theory and Notation}

We work in and assume basic familarity with homotopy type theory
(HoTT), i.e.\ Martin-L\"of type theory extended with higher inductive
types and the univalence axiom~\cite{hott-book}. The central concept
of HoTT is the Martin-L\"of identity type, which we write as $a = b$
--- we write $a \equiv b$ for definitional equality. We use Agda
notation $(x : A) \to B(x)$ for the type of dependent functions,
and write simply $A \to B$ if $B$ does not depend on $x : A$.
If the type
in the domain can be inferred from context, we may simply write
$\forall x. B(x)$ for $(x : A) \to B(x)$.
Freely occurring variables are assumed to be $\forall$-quantified.

We denote the type of dependent pairs by $\Sigma(x : A).B(x)$, and its
projections by $\fstproj$ and $\sndproj$. We write $A\times B$ if $B$ does not depend on $x:A$.
We write $\UU$ for a universe of types; we assume that we have a
cumulative hierarchy $\UU_i : \UU_{i+1}$ of such universes closed
under all type formers, but we will leave universe levels typically
ambiguous.

We call a type $A$ a \emph{proposition} if all elements of $A$ are
equal, i.e.\ if $(x : A) \to (y : A) \to x = y$ is provable. We write
$\Prop = \Sigma(A : \UU).\mathsf{isProp}(A)$ for the type of propositions, and
we implicitly insert a first projection if necessary, e.g.\ for
$A : \Prop$, we may write $x : A$ rather than $x : \fstproj(A)$.
A type $A$ is a \emph{set}, $A : \Set$, if $(x = y) : \Prop$ for every $x, y : A$.

By $\exists(x : A).B(x)$, we mean the \emph{propositional truncation}
of $\Sigma(x : A).B(x)$, 
and if
$(a, b) : \Sigma(x : A).B(x)$ then $|(a, b)| : \exists(x : A).B(x)$.
The elimination rule of $\exists(x : A).B(x)$ only allows to define
functions into propositions. By convention, we write $\exists k. P(k)$
for $\exists(k : \N).P(k)$.
Finally, we write $A \dissum B$ for the sum type, $\Empty$ for the empty type, $\Unit$ for the type
with exactly one element $\ast$, $\Bool$ for the type with two
elements $\bff$ and $\btt$,  and $\neg A$ for $A \to \Empty$.

The \emph{law of excluded middle} $(\LEM)$ says that, for every proposition $P$, we have $P \dissum \neg P$.
Since we explicitly work with constructive notions of ordinals, we do not assume $\LEM$, but rather use it as a \emph{taboo}: a statement is not provable constructively if it implies $\LEM$.
Another, weaker, constructive taboo is the
\emph{weak limited principle of omniscience} WLPO:
It says that any sequence $s : \N \to \Bool$ is either constantly $\mathsf{ff}$, 
or it is not constantly $\mathsf{ff}$. 

\section{Three Constructions of Types of Ordinals}
\label{sec:three-constructions-overview}

We consider three concrete notions of ordinals in this paper, together with their order relations $<$ and $\leq$.
The first notion is the one of \emph{Cantor normal forms}, written $\cnf$, whose order is decidable.
The second, written $\brouwer$, are \emph{Brouwer Trees}, implemented as a higher inductive-inductive type.
Finally, we consider the type $\bookord$ of ordinals that were studied in the HoTT book~\cite{hott-book}, whose order is undecidable, in general.
In the current section, we briefly give the three definitions and leave the discussion of results for afterwards.

\subsection{Cantor Normal Forms as a Subset of Binary Trees}
\label{subsec:CNF}

In classical set theory,
 every ordinal $\alpha$ can be written uniquely in Cantor normal form
\begin{equation}
  \label{eq:cnf-set-theory}
\alpha = \omega^{\beta_1} + \omega^{\beta_2} + \cdots + \omega^{\beta_n}
\;
\text{with }
\beta_1 \geq \beta_2 \geq \cdots \geq \beta_n
\end{equation}
for some natural number $n$ and ordinals $\beta_i$. If $\alpha < \varepsilon_0$,
then $\beta_i < \alpha$, and we can represent $\alpha$ as a finite binary tree (with a condition) as follows~\cite{buchholz:notation,CC:ord:coq,grimm:ord:coq,NFXG:three:ord}. Let $\TT$ be the type of unlabeled binary trees, i.e.\ the inductive type with suggestively named constructors $\tz : \TT$ and $\tom{-}{-} : \TT \times \TT \to \TT$.
Let the relation $<$ be the \emph{lexicographical order}, i.e.\ generated by the following clauses:
 \[
 \begin{aligned}
   & \tz < \tom a b \\
   & a < c \to \tom a b < \tom c d \\
   & b < d \to \tom a b < \tom a d.
  \end{aligned}
\]
We have the map $\fst:\TT \to \TT$ defined by $\fst(\tz) \defeq \tz$ and $\fst(\tom{a}{b}) \defeq a$ which gives us the left subtree (if it exists) of a tree.
A tree is a \emph{Cantor normal form} (CNF) if, for every $\tom{s}{t}$ that the tree contains, we have $\fst(t) \leq s$, where $s \leq t \defeq (s < t) \dissum (s = t)$; this enforces the condition in \eqref{eq:cnf-set-theory}. For instance, both trees $\tone \defeq \tom{\tz}{\tz}$ and $\omega \defeq \tom{\tone}{\tz}$ are CNFs.
Formally, the predicate $\mathsf{isCNF}$ is defined inductively by the two clauses
\[
 \begin{aligned}
  & \mathsf{isCNF}(\tz) \\
  & \mathsf{isCNF}(s) \to \mathsf{isCNF}(t) \to \fst(t) \leq s \to \mathsf{isCNF}(\tom{s}{t}).
 \end{aligned}
\]
We write
 $\cnf \defeq \Sigma(t : \TT). \mathsf{isCNF}(t)$
for the type of Cantor normal forms.
We often omit the proof of $\mathsf{isCNF}(t)$ and call the tree $t$ a CNF if no confusion is caused.

\subsection{Brouwer Trees as a Quotient Inductive-Inductive Type}
\label{subsec:Brouwer-as-HIIT}

As discussed in the introduction, \emph{Brouwer ordinal trees} (or simply \emph{Brouwer trees}) are in functional programming often inductively generated by the usual constructors of natural numbers (\emph{zero} and \emph{successor}) and a constructor that gives a Brouwer tree for every sequence of Brouwer trees.
To state a refined (\emph{correct} in a sense that we will make precise and prove) version, we need the following notions:

Let $A$ be a type and $\prec \, : A \to A \to \Prop$ be a binary relation.
If $f$ and $g$ are two sequences $\N \to A$, we say that $f$ is \emph{simulated by} $g$, written $f \precsim g$, if
$f \precsim g \defeq \forall k. \exists n. f(k) \prec g(n)$.
We say that $f$ and $g$ are \emph{bisimilar} with respect to $\prec$, written
\mbox{$f \approx^{\prec} g$},
if we have both $f \precsim g$ and $g \precsim f$.
A sequence $f : \N \to A$ is \emph{increasing} with respect to $\prec$
 if we have \mbox{$\forall k. f(k) \prec f(k+1)$}.
We write $\N \toincr{\prec} A$ for the type of $\prec$-increasing sequences. Thus an increasing sequence $f$ is a pair $f \equiv (\overline f, p)$ with $p$ witnessing that $\overline f$ is increasing, but we keep the first projection implicit and write $f(k)$ instead of $\overline f(k)$.

Our type of Brouwer trees is a \emph{quotient inductive-inductive type} \cite{Altenkirch2018},
where we  simultaneously construct the type $\brouwer : \Set$ together with a relation $\leq \,  : \brouwer \to \brouwer \to \Prop$.
The constructors for $\brouwer$ are
\[
\begin{alignedat}{3}
& \bzero &:&\;  &&\brouwer \\
& \bsuc &:& && \brouwer \to \brouwer \\
& \blimit &:& && (\N \toincr{<} \brouwer) \to \brouwer \\
& \bbisim & : & && (f \, g : \N \toincr{<} \brouwer) \to  f \approx^{\leq} g \to \blimit \, f = \blimit \, g,
\end{alignedat}
\]
where we write $x < y$ for $\bsuc \, x \leq y$ in the type of $\blimit$. Simulations thus use $\leq$ and the \emph{increasing} predicate uses $<$, as one would expect.
The truncation constructor, ensuring that $\brouwer$ is a set, is kept implicit in the paper (but is explicit in the Agda formalisation).

The constructors for $\leq$ are the following, where each constructor is implicitly quantified over the variables $x, y, z : \brouwer$ and $f : \N \toincr{<} \brouwer$ that it contains:
\[
\begin{alignedat}{3}
&\lzero &:& &\quad& 
    \bzero \leq x \\
&\ltrans &:& && 
    x \leq y \to y \leq z \to x \leq z\\
&\lsuccmono &:& && 
    x \leq y \to \bsuc \, x \leq \bsuc \, y \\
&\lcocone &:& && 
    (k : \N) \to x \leq f(k) \to x \leq \blimit \, f \\
&\llimiting &:& && 
    (\forall k. f(k) \leq x) \to \blimit \, f \leq x
\end{alignedat}
\]
The truncation constructor, which ensures that $x \leq y$ is a proposition, is again kept implicit.

We hope that the constructors of $\brouwer$ and $\leq$ are self-explanatory.
$\lcocone$ ensures that $\blimit \, f$ is indeed an upper bound of $f$, and $\llimiting$ witnesses that it is the \emph{least} upper bound or, from a categorical point of view, the (co)limit of $f$.

By restricting to limits of increasing sequences, we can avoid multiple representations of the same ordinal (as otherwise e.g.\ $a = \blimit \, (\lambda \_ . a)$). It is possible to drop this restriction, if one also strengthens the $\bbisim$ constructor to witness antisymmetry --- however we found this version of $\brouwer$ significantly harder to work with.

\subsection{Extensional Wellfounded Orders} \label{subsec:bookords}

The third notion of ordinals that we consider is the one studied in the HoTT book~\cite{hott-book}.
This is the notion which is closest to the classical definition of an ordinal as a set with a trichotomous, wellfounded, and transitive order, 
without a concrete representation.
Requiring trichotomy leads to a notion that makes many constructions impossible in a setting where the law of excluded middle is not assumed.
Therefore, when working constructively, it is better to replace the axiom of trichotomy by \emph{extensionality}.

Concretely, an ordinal in the sense of \cite[Def 10.3.17]{hott-book} is a type%
\footnote{Note that \cite[Def 10.3.17]{hott-book} asks for $X$ to be a set, but this follows from the rest of the definition and we therefore drop this requirement.}
$X$ together with a relation $\prec \, : X \to X \to \Prop$ which is \emph{transitive}, \emph{extensional} (any two elements of $X$ with the same predecessors are equal), and \emph{wellfounded} (every element is accessible, where accessibility is the least relation such that $x$ is accessible if every predecessor of $x$ is accessible.) --- we will recall the precise definitions in \cref{sec:axiomaticapproach}.
We write $\bookord$ for the type
of ordinals in this sense.
Note the shift of universes that happens here: the type $\bookord_i$ of ordinals with $X : \UU_i$ is itself in $\UU_{i+1}$. We are mostly interested in $\bookord_0$, but note that $\bookord_0$ lives in $\UU_1$, while $\cnf$ and $\brouwer$ both live in $\UU_0$.

We also have a relation on $\bookord$ itself.
Following \cite[Def 10.3.11 and Cor 10.3.13]{hott-book}, a \emph{simulation} between ordinals $(X,\prec_X)$ and $(Y,\prec_Y)$ is a function $f : X \to Y$ such that:
\begin{enumerate}[(a)]
    \item $f$ is monotone: $(x_1 \prec_X x_2) \to (f\, x_1 \prec_Y f\, x_2)$; and
    \item for all $x : X$ and $y : Y$, if $y \prec_Y f\, x$, then we have an $x_0 \prec_X x$ such that $f\, x_0 = y$. \label{item:simu-property}
    \end{enumerate}
We write $X \leq Y$ for the type of simulations between $(X,\prec_X)$ and $(Y,\prec_Y)$.
Given an ordinal $(X,\prec)$ and $x:X$, the \emph{initial segment} of elements below $x$ is given as
$X_{\slash x} \defeq \Sigma(y : X). y \prec x$.
Following \cite[Def 10.3.19]{hott-book},
a simulation
$f : X \leq Y$ is \emph{bounded} if we have $y : Y$ such that 
$f$ induces an equivalence $X \simeq Y_{\slash y}$.
We write $X < Y$ for the type of bounded simulations.
This completes the definition of $\bookord$ together with type families $\leq$ and $<$.


\section{An Abstract Axiomatic Framework for Ordinals} \label{sec:axiomaticapproach}

Which properties do we expect a type of ordinals to have?
In this section, we go up one level of abstraction.
We consider a type $A$ with 
type families $<$ and $\leq \, : A \to A \to \UU$, and discuss the properties that $A$ with $<$ and $\leq$ can have.
In \cref{sec:three-constructions-overview}, we introduced each of the types $\cnf$, $\brouwer$, and $\bookord$ together with its relations $<$ and $\leq$.
Note that $\leq$ is the reflexive closure of $<$ in the case of $\cnf$, but for $\brouwer$ and $\bookord$, this is not constructively provable.
In this section, we consider which properties they satisfy.

\subsection{General Notions}

$A$ is a \emph{set} if it satisfies the principle of unique identity proofs, i.e.\ if every identity type $a = b$ with $a, b: A$ is a proposition.
Similarly, $<$ and $\leq$ are \emph{valued in propositions} if every $a < b$ and $a \leq b$ is a proposition.
%
A relation $<$ is \emph{reflexive} if we have $a < a$, \emph{irreflexive} if it is pointwise not reflexive $\neg (a < a)$,
\emph{transitive} if $a < b \to b < c \to a < c$, and \emph{antisymmetric} if $a < b \to b < a \to a = b$.
Further, the relation $<$ is \emph{connex} if $(a < b) \dissum (b < a)$ and \emph{trichotomous} if $(a < b) \dissum (a = b) \dissum (b < a)$.

\begin{theorem} \label{thm:general-notions}
  Each of $\cnf$, $\brouwer$, and $\bookord$ is a set, and their relations $<$ and $\leq$ are all valued in propositions.
    In each case, both $<$ and $\leq$ are transitive, $<$ is irreflexive, and $\leq$ is reflexive and antisymmetric.
    For $\cnf$, the relation $<$ is trichotomous and $\leq$ connex; for $\bookord$, these statements are equivalent to the law of excluded middle.
\qedhref{https://cj-xu.github.io/agda/constructive-ordinals-in-hott/index.html\#1521}
  \end{theorem}

Proving that $\leq$ for $\brouwer$ is antisymmetric is challenging because of the path constructors in the inductive-inductive definition of Brouwer trees.
Antisymmetry and other technical properties discussed below require us to characterise the relation $\leq$ more explicitly, using an encode-decode argument~\cite{licataShulman_circle}.
By induction on $x$ and $y$, we define the family $\Code$ such that $(\Code \, x \, y) \leftrightarrow (x \leq y)$.
The cases for point constructors are unsurprising; for example, we define
\begin{align*}
  \Code \, (\blimit \, f) \,  (\bsuc \, y) &\defeq  \forall k. \Code \, (f \, k) \, (\bsuc \, y) \\
    \Code \, (\blimit \, f) \, (\blimit \, g) &\defeq
  \forall k. \exists n.  \Code \, (f \, k) \, (g \, n) \enspace .
\end{align*}
The difficult part is defining $\Code$ for the path constructor $\bbisim$.
If for example we have $g \approx h$, we need to show that $\Code \; (\blimit \, f) \, (\blimit \, g) \; = \; \Code \; (\blimit \, f) \, (\blimit \, h)$.
The core argument is easy; using the bisimulation $g \approx h$, one can translate between indices for $g$ and $h$ with the appropriate properties.
However, this example already shows why this becomes tricky: The bisimulation gives us inequalities ($\leq$), but the translation requires instances of $\Code$, which means that $\toCode : (x \leq y) \to (\Code \, x \, y)$ has to be defined \emph{mutually} with $\Code$.
This is still not sufficient: In total, the mutual higher inductive-inductive construction needs to simultaneously prove and construct $\Code$, $\toCode$, versions of transitivity and reflexivity of $\Code$ as well several auxiliary lemmas.
The complete definition is presented in the Agda formalisation (file \href{https://cj-xu.github.io/agda/constructive-ordinals-in-hott/BrouwerTree.Code.html}{\texttt{BrouwerTree.Code}}).
Once the definition of $\Code$ is shown correct, many technical properties are simple consequences.


From now on, we will assume that $A$ is a set and that $<$ and $\leq$ are valued in propositions.


\subsection{Extensionality and Wellfoundedness}


Following~\cite[Def~10.3.9]{hott-book}, we call a relation $<$ \emph{extensional} if, for all $a, b : A$, we have
$(\forall c. c < a \leftrightarrow c < b) \to b = a$,
where $\leftrightarrow$ denotes ``if and only if'' (functions in both directions). Extensionality of $<$ for $\brouwer$ is true, but non-trivial -- note that it fails for the ``naive'' version of $\brouwer$, where the path constructor $\bbisim$ is missing.

\begin{theorem} \label{thm:<-extensional}
  For each of $\cnf$, $\brouwer$, $\bookord$, both $<$ and $\leq$ are
  extensional.
  \qedhref{https://cj-xu.github.io/agda/constructive-ordinals-in-hott/index.html\#2511}
\end{theorem}

We use the inductive definition of accessibility and wellfoundedness (with respect to $<$) by Aczel~\cite{aczelinductive}.
Concretely, the type family
$\acc : A \to \UU$ is inductively defined by the constructor
\[
\mathit{access} : (a : A) \to ((b : A) \to b < a \to \acc(b)) \to \acc(a).
\]
An element $a : A$ is called \emph{accessible} if $\acc(a)$, and $<$ is \emph{wellfounded} if all elements of $A$ are accessible. It is well known that the following induction principle can be derived from the inductive presentation~\cite{hott-book}:

\begin{lemma}[Transfinite Induction] \label{lm:wf:ti}
    Let $<$ be wellfounded and $P : A \to \UU$ be a type family such that
    $\forall a. (\forall b<a.P(b)) \to P(a)$.
    Then, it follows that $\forall a. P(a)$.
    \qedfullhref{https://cj-xu.github.io/agda/constructive-ordinals-in-hott/index.html\#2828}
\end{lemma}

In turn, transfinite induction can be used to prove that there is no infinite decreasing sequence if $<$ is wellfounded:
  $\neg \left( \Sigma (f : \N \to A). (i : \N) \to f(i+1) < f(i) \right)$.
A direct corollary is that if $<$ is wellfounded and valued in propositions, then its reflexive closure $(x < y) \dissum (x = y)$ is also valued in propositions, as $b < a$ and $b = a$ are mutually exclusive propositions.


\begin{theorem} \label{thm:<-wellfounded}
  For each of $\cnf$, $\brouwer$, $\bookord$, the relation $<$ is wellfounded.
  \qedhref{https://cj-xu.github.io/agda/constructive-ordinals-in-hott/index.html\#2993}
\end{theorem}

The proof for $\brouwer$ again makes crucial use of our encode-decode characterisation of $\leq$. Whenever $x < \blimit\,f$, we can use the characterisation to find an $n : \N$ such that $x < f(n)$, which allows us to proceed with an inductive proof of wellfoundedness.
Note that the results stated so far in particular mean that $\cnf$ and $\brouwer$ can be seen as elements of $\bookord$ themselves.


\subsection{Classification as Zero, a Successor, or a Limit}

All standard formulations of ordinals allow us to determine a minimal ordinal \emph{zero} and (constructively) calculate the \emph{successor} of an ordinal, but only some allows us to also calculate the \emph{supremum} or \emph{limit} of a collection of ordinals.

\subsubsection{Assumptions}
\label{subsubsec:abstract-assumption}
We have so far not required a relationship between $<$ and $\leq$, but we now need to do so in order for the concepts we define to be meaningful.
We assume:
\begin{enumerate}[({A}1)]
    \item \label{item:assumption1}
    $<$ is transitive and irreflexive;
    \item \label{item:assumption2}
    $\leq$ is reflexive, transitive, and antisymmetric;
    \item \label{item:assumption3}
    we have $(<) \subseteq (\leq)$ and $(< \circ \leq) \subseteq (<)$.
\end{enumerate}
The third condition \ref{item:assumption3} means that $(b < a) \to (b \leq a)$ and $(c < b) \to (b \leq a) \to (c < a)$.
The ``symmetric'' variation
\[
(\leq \circ <) \subseteq (<)
\]
is true for $\cnf$ and $\brouwer$, but for $\bookord$, it is
equivalent to the law of excluded middle --- hence, we do not assume
it. This constructive failure is known, and can be seen as motivation
for \emph{plump} ordinals~\cite{taylor:ordinals,shulman:plump}.
Of course, the above assumptions are satisfied if $\leq$ is the reflexive closure of $<$, but we again emphasise that this is not necessarily the case.

\begin{theorem} \label{thm:all-satisfy-assumptions}
    For each of $\cnf$, $\brouwer$, $\bookord$, assumptions \ref{item:assumption1} to \ref{item:assumption3} are satisfied.
    \qedhref{https://cj-xu.github.io/agda/constructive-ordinals-in-hott/index.html\#3153}
\end{theorem}

For the remaining concepts, we assume that $<$ and $\leq$ satisfy the discussed assumptions.

\subsubsection{Zero and (Strong) Successors}

Let $a$ be an element of $A$.
It is \emph{zero}, or \emph{bottom}, if it is at least as small as any other element
\begin{equation}
  \label{eq:iszero}
\iszero(a) \defeq \forall b. a \leq b,
\end{equation}
and we say that the triple $(A,<,\leq)$ \emph{has a zero} if we have an inhabitant of the type $\Sigma(z : A). \iszero(z)$. Both the types ``being a zero'' and ``having a zero'' are propositions.

\begin{theorem} \label{thm:all-zero}
  $\cnf$, $\brouwer$, $\bookord$ each have a zero.
  \qedhref{https://cj-xu.github.io/agda/constructive-ordinals-in-hott/index.html\#3561}
\end{theorem}

We say that $a$ is a \emph{successor} of $b$ if it is the least element strictly greater\footnote{Note that $>$ and $\geq$ are the obvious symmetric notations for $<$, $\leq$; they are \emph{not} newly assumed relations.} than $b$:
\[
  (\isupsucof a b) \defeq (b < a) \times \forall x > b. x \geq a.
\]
We say that $(A,<,\leq)$ \emph{has successors} if
there is a function $s : A \to A$ which calculates successors, i.e.\ such that $\forall b. \issucof{s(b)}{b}$. ``Calculating successors'' and ``having successors'' are propositional properties, i.e.\ if a function that calculates successors exists, then it is unique. The following statement is simple but useful. Its proof uses assumption \ref{item:assumption3}.
\begin{lemma}
  \label{thm:calc-succ-characterisation}
    Let $s : A \to A$ be given. The function
    $s$ calculates successors if and only if $\forall b x. (b < x) \leftrightarrow (s\, b \leq x)$.
    \qedfullhref{https://cj-xu.github.io/agda/constructive-ordinals-in-hott/index.html\#4080}
\end{lemma}

Dual to ``$a$ is the least element strictly greater than $b$'' is the 
statement that ``$b$ is the greatest element strictly below $a$'', in which case it is natural to call $b$ the \emph{predecessor} of $a$.
If $a$ is the successor of $b$ and $b$ the predecessor of $a$, then we call $a$ the \emph{strong successor} of $b$:
\[
\isstrongsucof{a}{b} \defeq \issucof{a}{b} \times \forall x < a. x \leq b.
\]
We say that $A$ \emph{has strong successors} if there is $s : A \to A$ which calculates strong successors, i.e.\ such that $\forall b. \isstrongsucof{s(b)}{b}$.
The additional information contained in a strong successor play an important role in our technical development. 
A function $f : A \to A$ is \emph{$<$-monotone} or \emph{$\leq$-monotone} if it preserves the respective relation.

\begin{theorem}\label{thm:strong-succs}
  Each of the three types $\cnf$, $\brouwer$, $\bookord$ has strong successors.
  The successor functions of $\cnf$ and $\brouwer$ are both $<$- and $\leq$-monotone.
  For the successor function of $\bookord$, either monotonicity property is equivalent to the law of excluded middle.
  \qedhref{https://cj-xu.github.io/agda/constructive-ordinals-in-hott/index.html\#4303}
\end{theorem}

For $\cnf$, the successor function is given by adding a leaf, for $\brouwer$ by the constructor with the same name, and for $\bookord$, one forms the coproduct with the unit type.
%

\subsubsection{Suprema and Limits}

Finally, we consider \emph{suprema/least upper bounds} of $\N$-indexed sequences.
%
We say that $a$ is a \emph{supremum} or the \emph{least upper bound} of $f : \N \to A$, if $a$ is at least as large as every $f_i$, and if any other $x$ with this property is at least as large as $a$:
\[
(\isupsupof a f) \defeq (\forall i. f_i \leq a) \times (\forall x. (\forall i. f_i \leq x) \to a \leq x). 
\]
We say that $(A,<,\leq)$ \emph{has suprema} if there is a function $\sqcup : (\N \to A) \to A$ which calculates suprema, i.e.\ such that $(f : \N \to A) \to \issupof{(\sqcup f)}{f}$.  The supremum of a sequence is unique if it exists, i.e.\ the type of suprema is propositional for a given sequence $f$.
Both the properties ``calculating suprema'' and ``having suprema'' are propositions.

Every $a : A$ is trivially the supremum of the sequence constantly $a$, and therefore, ``being a supremum'' does not describe the usual notion of \emph{limit ordinals}.
One might consider $a$ a \emph{proper} supremum of $f$ if $a$ is pointwise strictly above $f$, i.e.\ $\forall i. f_i < a$.
This is automatically guaranteed if $f$ is
increasing with respect to $<$, 
and in this case, we call $a$ the \emph{limit} of $f$:
\[
\begin{aligned}
& \islimof \blank \blank : A \to (\N \toincr{<} A) \to \UU \\
& \islimof{a}{(f,q)} \defeq \issupof{a}{f}.
\end{aligned}
\]
We say that $A$ \emph{has limits} if there is a function $\mathsf{limit} : (\N \toincr{<} A) \to A$ that calculates limits.

Note that $\cnf$ cannot have limits since one can construct a sequence (see \cref{thm:cnf-below-eps0}) which comes arbitrarily close to $\varepsilon_0$.
This motivates the restriction to \emph{bounded} sequences, i.e.\ a sequence $f$ with a $b : \cnf$ such that $f_i < b$ for all $i$.

\begin{theorem} \label{thm:sups-lims}
  $\cnf$ does not have suprema or limits.
  $\brouwer$ has limits of increasing sequences by construction.
  $\bookord$ also has limits of increasing sequences, and moreover limits of \emph{weakly} increasing sequences (i.e.\ sequences increasing with respect to $\leq$).
  
  Assuming the law of exclude middle, $\cnf$ has suprema (and thus limits) of arbitrary \emph{bounded} sequences. If $\cnf$ has limits of bounded increasing sequences, then the weak limited principle of omniscience (WLPO) is derivable.
  \qedhref{https://cj-xu.github.io/agda/constructive-ordinals-in-hott/index.html\#4750}
\end{theorem}
We expect that it is not constructively possible to calculate suprema (or even binary joins) in $\brouwer$, as it seems this would make it possible to decide if a limit reaches past $\omega + 1$ or not, which is a constructive taboo.

%

%
%
%

\subsubsection{Classifiability}

For classical set-theoretic ordinals, every ordinal is either zero, a successor, or a limit.
We say that a notion of ordinals which allows this is has classification.
This is very useful, as many theorems that start with ``for every ordinal'' have proofs that consider the three cases separately.
In the same way as not all definitions of ordinals make it possible to calculate limits, only some formulations make it possible to constructively classify any given ordinal.
We already defined what it means to be a zero in \eqref{eq:iszero}.
We now also define what it means for $a : A$ to be a strong successor or a limit:
\[
  \isstrongsuc(a) \defeq \Sigma(b:A).  (\isstrongsucof a b)
  \qquad\quad
  \islim(a) \defeq \exists f : \N \to A. \islimof{a}{f}.
\]
All of $\iszero(a)$, $\isstrongsuc(a)$ and $\islim(a)$ are propositions; note that this is true even though $\isstrongsuc(a)$ is defined without a propositional truncation.

\begin{lemma}
  \label{lem:only-one-out-of-three}
    Any $a:A$ can be at most one out of \{zero, strong successor, limit\}, and in a unique way.
    In other words, the type $\iszero(a) \dissum \isstrongsuc(a) \dissum \islimit(a)$ is a proposition.
    \qedfullhref{https://cj-xu.github.io/agda/constructive-ordinals-in-hott/index.html\#5485}
\end{lemma}
    

We say that an element of $A$ is \emph{classifiable}
if it is zero or a strong successor or a limit.
We say $(A,<,\leq)$ \emph{has classification} if every element of $A$ is classifiable.
By \cref{lem:only-one-out-of-three}, $(A, <, \leq)$ has classification exactly if the type $\iszero(a) \dissum \isstrongsuc(a) \dissum \islimit(a)$
is contractible.

\begin{theorem}
  \label{thm:classification}
  $\cnf$ and $\brouwer$ have classification.
  $\bookord$ having classification would imply the law of excluded middle.
  \qedhref{https://cj-xu.github.io/agda/constructive-ordinals-in-hott/index.html\#5697}
\end{theorem}

Classifiability corresponds to a case distinction, but the useful
principle from classical ordinal theory is the related induction principle:

\begin{definition}[classifiability induction]
    \label{def:CFI}
    We say that $(A,<,\leq)$ satisfies the principle of \emph{classifiability induction} if the following holds:
    For every family $P : A \to \Prop$ such that
    \begin{align*}
    & \iszero(a) \to P(a) \\
    & (\isstrongsucof{a}{b}) \to P(b) \to P(a) \\
    & (\islimof a f) \to (\forall i. P(f_i)) \to P(a),
    \end{align*}
    we have $\forall a. P(a)$.
\end{definition}

Note that classifiability induction does \emph{not} ask for successors or limits to be computable. Using \cref{lem:only-one-out-of-three}, we get that classifiability induction implies classification.
For the reverse, we need a further assumption:
\begin{theorem}
    \label{thm:CFI}
    Assume $(A,<,\leq)$ has classification and satisfies the principle of transfinite induction.
    Then $(A,<,\leq)$ satisfies the principle of classifiability induction.
    \qedfullhref{https://cj-xu.github.io/agda/constructive-ordinals-in-hott/index.html\#6376}
\end{theorem}

It is also standard in classical set theory that classifiability induction implies transfinite induction:
showing $P$ by transfinite induction corresponds to showing $\forall x<a. P(x)$ by classifiability induction.
In our setting, this would require strong additional assumptions, including the assumption that $(x \leq a)$ is equivalent to $(x < a) \dissum (x = a)$, i.e.\ that $\leq$ is the reflexive closure of $<$.
The standard proof works with several strong assumptions of this form, but we do not consider this interesting or useful, and concentrate on the results which work for the weaker assumptions that are satisfied for $\brouwer$ and $\bookord$ (see \cref{subsubsec:abstract-assumption}).

\begin{theorem}
  \label{thm:classifyability-induction}
  $\cnf$ and $\brouwer$ satisfy classifiability induction, while $\bookord$ satisfying it again implies excluded middle.
  \qedhref{https://cj-xu.github.io/agda/constructive-ordinals-in-hott/index.html\#6694}
  \end{theorem}

\subsection{Arithmetic}
\label{subsec:abstract-arithmetic}

Using the predicates $\iszero(a)$, $\issucof a b$, and $\issupof{a}{f}$, we can define what it means for $(A,<,\leq)$ to have the standard arithmetic operations.
We still work under the assumptions declared in \cref{subsubsec:abstract-assumption} --- in particular, we do not assume that e.g.\ limits can be calculated, which is important to make the theory applicable to $\cnf$.

\begin{definition}[having addition] \label{def:have-addition}
    We say that $(A,<,\leq)$ \emph{has addition}
    if there is a function $+ : A \to A \to A$ which satisfies the following properties:
    \begin{align}
    & \iszero(a) \to c + a = c \notag \\
    & \issucof a b \to \issucof {d} {(c+b)} \to c + a = d \notag \\
    & \islimof{a}{f} \to \issupof b {(\lambda i. c + f_i)} \to c + a = b  \label{eq:having-addition-3}
    \end{align}
    We say that $A$ \emph{has unique addition} if there is exactly one function $+$ with these properties.
\end{definition}

Note that \eqref{eq:having-addition-3} makes an assumption only for (strictly) \emph{increasing} sequences $f$; this suffices to define a well-behaved notion of addition, and it is not necessary to include a similar requirement for arbitrary sequences.
Since $(\lambda i. c + f_i)$ is a priori not necessarily increasing, the middle term of \eqref{eq:having-addition-3} has to talk about the supremum, not the limit.

Completely analogously to addition, we can formulate multiplication and exponentation, again without assuming that successors or limits can be calculated:

\begin{definition}[having multiplication] \label{def:have-multiplication}
    Assuming that $A$ has addition, we say that it \emph{has multiplication}
    if we have a function $\cdot : A \to A \to A$ that satisfies the following properties:
    \begin{align*}
    & \iszero(a) \to c \cdot a = a \\
    & \issucof a b \to c \cdot a = c \cdot b + c \\
    & \islimof{a}{f} \to \issupof b {(\lambda i. c \cdot f_i)} \to c \cdot a = b
    \end{align*}
    $A$ \emph{has unique multiplication} if it has unique addition and there is exactly one function $\cdot$ with the above properties.
\end{definition}

\begin{definition}[having exponentation]
\label{def:have-exponentation}
    Assume $A$ has addition and multiplication.
    We say that $A$ \emph{has exponentation with base $c$}
    if we have a function $\mathsf{exp}(c,-) : A \to A$ that satisfies the following properties:
    \begin{align*}
    & \iszero(b) \to \issucof a b \to \mathsf{exp}(c,b) = a \\
    & \issucof a b \to \mathsf{exp}(c,a) = \mathsf{exp}(c,b) \cdot c \\
    & \begin{multlined}
    \islimof{a}{f} \to \neg \iszero(c) \to {\issupof{b}{(\mathsf{exp}(c,f_i))} \to \mathsf{exp}(c,a) = b}
    \end{multlined} \\
    & \islimof{a}{f} \to \iszero(c) \to \mathsf{exp}(c,a) = c
    \end{align*}
    $A$ \emph{has unique exponentation with base $c$} if it has unique addition and multiplication, and if $\mathsf{exp}(c, -)$ is unique.
\end{definition}

\begin{theorem} \label{thm:all-arithmetic-operations}
$\cnf$ has addition, multiplication, and exponentiation with base $\omega$ (all unique), $\brouwer$ has addition, multiplication and exponentiation with every base (all unique), and $\bookord$ has addition and multiplication.
\qedhref{https://cj-xu.github.io/agda/constructive-ordinals-in-hott/index.html\#7417}
\end{theorem}
For $\cnf$, arithmetic is defined by pattern matching on the trees.
Addition\footnote{Caveat: $\fatplus$ is a notation for the tree constructor, while $+$ is an operation that we define. We use parenthesis so that all operations can be read with the usual operator precedence.} is given as
\begin{alignat*}{3}
\tz + b & \defeq b \\
a + \tz & \defeq a \\
(\tom{a}{c}) + (\tom{b}{d}) & \defeq
\begin{cases}
\tom{b}{d} & \text{if $a<b$} \\
\tom{a}{(c+(\tom{b}{d}))} & \text{otherwise,}
\end{cases}
\end{alignat*}
multiplication as
\begin{alignat*}{3}
\tz \cdot b & \defeq \tz \\
a \cdot \tz & \defeq \tz \\
a \cdot (\tom{\tz}{d}) & \defeq a + a \cdot d \\
(\tom{a}{c}) \cdot (\tom{b}{d}) & \defeq (\tom{a+b}{\tz}) + (\tom{a}{c}) \cdot d \quad\text{if $b \not= \tz$,}
\end{alignat*}
and exponentiation with base $\upomega$ by $\upomega^a \defeq \tom a \tz$.
These definitions are standard.
Novel is our proof of correctness in the sense of \cref{def:have-addition,def:have-multiplication,def:have-exponentation}, which we achieve by defining the inverse operations of subtraction and division.

Arithmetic on $\brouwer$ is defined by recursion on the second argument, following the clauses of the specifications in \cref{def:have-addition,def:have-multiplication,def:have-exponentation}.
Since the constructor $\blimit$ only accepts an increasing sequence, it is necessary to
prove mutually with the definition that the operations are monotone and preserve increasing sequences.
However, the case $c = 0$ needs to be treated separately since neither pointwise multiplication nor exponentiation with $0$ preserves increasingness.
This makes it crucial to use classification (\cref{thm:classification}) and, in particular, that it is decidable whether $c : \brouwer$ is zero.

Addition on $\bookord$ is given by disjoint union $A \dissum B$ (with $\inl(a) \prec_{A \dissum B} \inr(b)$), and multiplication by Cartesian product $A \times B$ with the reverse lexicographical order. We expect that exponentation cannot be defined constructively: the ``obvious'' definition via function spaces gives a wellfounded order assuming the law of excluded middle~\cite{holland2005lexicographic}, but it seems unlikely that it can be avoided.

\section{Interpretations Between the Notions}
\label{sec:interpretations}

In this section, we show how our three notions of ordinals can be connected via structure preserving embeddings.

\subsection{From Cantor Normal Forms to Brouwer Trees}

The arithmetic operations of $\brouwer$ allow the construction of a function $\CtoB : \cnf \to \brouwer$ in a canonical way.
We define $\CtoB : \cnf \to \brouwer$ by:
\begin{align*}
  \CtoB(\tz) &\defeq \bzero \\
  \CtoB(\tom a b) &\defeq \omega^{\CtoB(a)} + \CtoB(b)
\end{align*}

\begin{theorem} \label{thm:CtoB-reflects}
  The function $\CtoB$ preserves and reflects $<$ and $\leq$, i.e., $a < b \iff \CtoB(a) < \CtoB(b)$, and  $a \leq b \iff \CtoB(a) \leq \CtoB(b)$.
  \qedfullhref{https://cj-xu.github.io/agda/constructive-ordinals-in-hott/index.html\#7764}
\end{theorem}
%
%

To show that $\CtoB$ preserves $<$, we first prove that Brouwer trees of the form $\omega^x$ are additive principal: if $a < \omega^x$ then $a + \omega^x = \omega^x$ --- a property not true for the ``naive'' version of Brouwer trees without path constructors.
By reflecting $\leq$ and antisymmetry, we have:

\begin{corollary}\label{cor:CtoB-injective}
  The function $\CtoB$ is injective.
  \qedfullhref{https://cj-xu.github.io/agda/constructive-ordinals-in-hott/index.html\#7978}
\end{corollary}

We note that $\CtoB$ also preserves all arithmetic operations
on $\cnf$. For multiplication, this relies on $\iota(n) \cdot \omega^x = \omega^x$ for $\brouwer$, where $\iota : \N \to \brouwer$ embeds the natural numbers as Brouwer trees, and  $\omega \defeq \blimit\, \iota$ --- see our formalisation for details.

\begin{theorem}
    \label{lem:f-arith}
    $\CtoB$ commutes with addition, multiplication, and exponentiation with base $\omega$.
    \qedfullhref{https://cj-xu.github.io/agda/constructive-ordinals-in-hott/index.html\#8063}
\end{theorem}

Lastly, as expected, Brouwer trees define bigger ordinals than Cantor normal forms: when embedded into $\brouwer$, all Cantor normal forms are below $\varepsilon_0$, the limit of the increasing sequence $\omega$, $\omega^{\omega}$, $\omega^{\omega^{\omega}}$, \ldots

\begin{theorem}\label{thm:cnf-below-eps0}
  For all $a : \cnf$, we have $\CtoB(a) < \blimit\,(\lambda k.\omega \uparrow\uparrow k)$, where $\omega \uparrow\uparrow 0 \defeq \omega$ and $\omega \uparrow\uparrow (k+1) \defeq \omega^{\omega \uparrow\uparrow k}$.
  \qedfullhref{https://cj-xu.github.io/agda/constructive-ordinals-in-hott/index.html\#8359}
\end{theorem}

%

\subsection{From Brouwer Trees to Extensional Wellfounded Orders}

As $\brouwer$ comes with an order that is extensional, wellfounded, and transitive, it can itself be seen as an element of $\bookord$.
Every ``subtype'' of $\brouwer$ (constructed by restricting to trees smaller than a given tree) inherits this property, giving a canonical function from Brouwer trees to extensional, wellfounded orders. We define
\[
\BtoO(a) = \Sigma(y : \brouwer).(y < a).
\]
with order relation $(y, p) \prec (y', p')$ if $y < y'$. This extends to a function $\BtoO : \brouwer \to \bookord$.
The first projection gives a simulation $\BtoO(a) \leq \brouwer$. Using extensionality of $\brouwer$, this implies that $\BtoO$ is an embedding from $\brouwer$ into $\bookord$. Using that $<$ on $\brouwer$ is propositional, and that
carriers of orders
are sets, it is also not hard to see that $\BtoO$ is
order-preserving:

\begin{lemma}\label{lem:BtoO-injective}
  The function $\BtoO : \brouwer \to \bookord$ is injective, and preserves  $<$ and $\leq$.
  \qedhref{https://cj-xu.github.io/agda/constructive-ordinals-in-hott/index.html\#8420}
\end{lemma}

A natural question is whether the above result can be strengthened further, i.e.\ whether $\BtoO$ is a simulation.
Using $\LEM$ to find a minimal simulation witness, this is possible:

\begin{theorem} \label{thm:lem-implies-simulation}
  Under the assumption of the law of excluded middle, the function $\BtoO : \brouwer \to \bookord$ is a simulation.
  \qed
\end{theorem}

We do not know whether the reverse of \cref{thm:lem-implies-simulation} is provable, but from the assumption that $\BtoO$ is a simulation, we can derive another constructive taboo:

\begin{theorem}
  \label{thm:BtoO-sim-WLPO}
  If the map $\BtoO : \brouwer \to \bookord$ is a simulation, then WLPO holds.
  \qed
\end{theorem}

We trivially have $\BtoO(\bzero) = \Empty$. 
One can further prove that $\BtoO$ commutes with limits, i.e.\ 
$\BtoO(\blimit(f)) = \osup{(\BtoO \circ f)}$.
However, $\BtoO$ does \emph{not} commute with successors;
it is easy to see that $\BtoO \, x \dissum \Unit \leq \BtoO(\bsuc \, x)$, but the other direction 
implies WLPO.
This also means that $\BtoO$ does not preserve the arithmetic operations but ``over-approximates''
them, i.e.\ we have $\BtoO(x + y) \geq \BtoO \, x \dissum \BtoO \, y$
and
$\BtoO(x \cdot y) \geq \BtoO \, x \times \BtoO \, y$.

\section{Conclusions and Future Directions}

We have demonstrated that three very different implementations of ordinal numbers, namely Cantor normal forms $(\cnf)$, Brouwer ordinal trees $(\brouwer)$, and extensional wellfounded orders $(\bookord)$, can be studied in a single abstract setting in the context of homotopy type theory.
We hope that our development may shed light on other constructive or formalised approaches to ordinals also in other settings~\cite{MV:ord:acl2,BPT:card:isa,BFT:nested:mset,schmitt:ord:key}.

Cantor normal forms are a formulation where most properties are decidable, while the opposite is the case for extensional wellfounded orders.
Brouwer ordinal trees sit in the middle, with some of its properties being decidable. 
This aspect is not discussed in full in this paper; we only have included \cref{thm:classifyability-induction}.
It is easy to see that, for $x : \brouwer$, it is decidable whether $x$ is finite; in other words, the predicate $(\omega \leq \blank) : \brouwer \to \Prop$ is decidable,
while  $(\omega < \blank)$ is decidable if and only if WLPO holds.

If $x$ is finite, then the predicates $(x = \blank)$, $(x \leq \blank)$, and $(x < \blank)$ are also decidable.
We have a further proof that, if $c : \cnf$ is smaller than $\omega^2$,
then the families $(\CtoB \, c \leq \blank)$ and $(\CtoB \, c < \blank)$ are \emph{semidecidable}, where semidecidability can be defined using the \emph{Sierpinski space}~\cite{alt-dan-kra:partiality,chapman2019quotienting,veltri2017type}.

Thus, each of the canonical maps $\CtoB : \cnf \to \brouwer$ and $\BtoO : \brouwer \to \bookord$ embeds the ``more decidable'' formulation of ordinals into the ``less decidable'' one.
Naturally, they both also include a ``smaller'' type of ordinals into a ``larger'' one:
While every element of $\cnf$ represents an ordinal below $\epsilon_0$, $\brouwer$ can go much further.
It would be interesting to consider more powerful ordinal notation systems such as those based on the Veblen functions~\cite{veblen,schuette} or collapsing functions\cite{bachmann,buchholz}, and see how these compare to Brouwer trees.
Another avenue for potentially extending Cantor normal forms would be using \emph{superleaves}~\cite{Dershowitz_superleaves}; we do not know how such a ``bigger'' version of $\cnf$ would compare to $\brouwer$.

%

%
Since $\brouwer$ can be viewed as an element of $\bookord$, the latter can clearly reach larger ordinals than the former.
This is of course not surprising; the Burali-Forti argument~\cite{escardo_et_al:BuraliForti,burali1897questione} shows that lower universes cannot reach the same ordinals as higher universes. Another obstruction for $\brouwer$ to reach the full power of $\bookord$ is the fact that $\brouwer$ only includes limits of $\N$-indexed sequences. To overcome this problem, one can similarly construct higher number classes as quotient inductive-inductive types, e.g.\ a type $\brouwer_3$ closed under limits of $\brouwer$-indexed sequences, and then more generally types $\brouwer_{n+1}$ closed under limits of $\brouwer_n$-indexed sequences, and so on.

Finally, there are interesting connections between the ordinals we can
represent and the proof-theoretic strength of the ambient type theory:
each proof of wellfoundedness for a system of ordinals is also a lower
bound for the strength of the type theory it is constructed in. It is
well known that definitional principles such as simultaneous
inductive-recursive definitions~\cite{IRstrength} and higher inductive
types~\cite{lumsdaine:hits} can increase the proof-theoretical
strength, and so, we hope that they can also be used to faithfully
represent even larger ordinals.

%


\bibliographystyle{plain}
\bibliography{references}

\appendix

\section*{Appendix}
\label{sec:appendix}
\addcontentsline{toc}{section}{Appendix}

In addition to our Agda formalisation,
we include paper proofs of the results listed in the paper, organised as follows:
\begin{itemize}
	\item \cref{app:general} contains arguments for the statements that are formulated in our abstract framework;
	\item \cref{app:cnf} proves the results for Cantor normal forms;
	\item \cref{app:brouwer} proves the theorems about Brouwer trees;
	\item \cref{app:orders} gives the arguments for extensional wellfounded orders;
	\item and \cref{app:interpretations} contains proofs for \cref{sec:interpretations}.
\end{itemize}

\section{General proofs}
\label{app:general}

\begingroup
\def\thelemma{\ref{lm:wf:ti}}
\begin{lemma}[Transfinite Induction] 
	Let $<$ be wellfounded and $P : A \to \UU$ be a type family such that
	$\forall a. (\forall b<a.P(b)) \to P(a)$.
	Then, it follows that $\forall a. P(a)$.
\end{lemma}
\begin{proof}
	This is a special case of the induction principle that one gets from the inductive definition of the accessibility family.
\end{proof}
\addtocounter{theorem}{-1}
\endgroup

\begingroup
\def\thelemma{\ref{thm:calc-succ-characterisation}}
\begin{lemma}
	Let $s : A \to A$ be given. The function
	$s$ calculates successors if and only if $\forall b x. (b < x) \leftrightarrow (s\, b \leq x)$.
\end{lemma}
\begin{proof}
    Assume that $s$ calculates successors; we have to show  $(b < x) \leftrightarrow (s\, b \leq x)$. The part $\to$ is clear.
    Further, since $s\, b$ is the successor of $b$, we have $b < s\, b$. Together with
    $(<) \subseteq (\leq)$ and $(< \circ \leq) \subseteq (<)$, this shows the part $\leftarrow$.

    Now assume $(b < x) \leftrightarrow (s\, b \leq x)$.
    By reflexivity of $\leq$, the case $x \equiv s\, b$ shows $b < s\, b$, and the direction $\to$ directly gives the second part of ``$s$ calculates successors''.
\end{proof}
\addtocounter{theorem}{-1}
\endgroup

\begingroup
\def\thelemma{\ref{lem:only-one-out-of-three}}
\begin{lemma}
	Any $a:A$ can be at most one out of \{zero, strong successor, limit\}, and in a unique way.
	In other words, the type $\iszero(a) \dissum \isstrongsuc(a) \dissum \islimit(a)$ is a proposition.
\end{lemma}
\begin{proof}
    First part: Each of the three summands is individually a proposition.
    This is clear for $\iszero(a)$ and $\islimit(a)$.
    Assume that $a$ is the strong successor of both $b$ and $b'$; then we have $b \leq b'$ and $b' \leq b$, implying $b = b'$ by antisymmetry.
    
    Second part: We now only have to check that any two summands exclude each other.
    Since the goal is a proposition, we can assume that we are given $b$ and $f$ in the successor and limit case.
    Assume that $a$ is zero and the successor of $b$. This implies $b < a \leq b$ and thus $b < b$, contradicting irreflexivity.
    If $a$ is zero and the limit of $f$, the same argument (with $b$ replaced by $f_0$) applies.
    Finally, assume that $a$ is the strong successor of $b$ and the limit of $f$. These assumptions show that $b$ is an upper bound of $f$, thus we get $a \leq b$. Together with $b < a$, this gives the contradiction $b < b$ as above.
\end{proof}
\addtocounter{theorem}{-1}
\endgroup

\begingroup
\def\thetheorem{\ref{thm:CFI}}
\begin{theorem}
	Assume $(A,<,\leq)$ has classification and satisfies the principle of transfinite induction.
	Then $(A,<,\leq)$ satisfies the principle of classifiability induction:
        for every family $P : A \to \Prop$ such that
            \begin{align}
    & \iszero(a) \to P(a) \label{eq:class-ind-zero} \\
    & (\isstrongsucof{a}{b}) \to P(b) \to P(a) \label{eq:class-ind-suc}\\
    & (\islimof a f) \to (\forall i. P(f_i)) \to P(a), \label{eq:class-ind-lim}
    \end{align}
    we have $\forall a . P(a)$.
\end{theorem}
\begin{proof}
	By transfinite induction, it suffices to show
	\begin{equation} \label{eq:tfi-ass}
	(\forall b<a. P(b)) \to P(a)
	\end{equation}
	for some fixed $a$.
	By classification, we can consider three cases.
	If $\iszero(a)$, then \eqref{eq:class-ind-zero} gives us $P(a)$, which shows \eqref{eq:tfi-ass} for that case.
	If $a$ is the strong successor of $b$, we use that the predecessor $b$ is one of the elements that 
	the assumption of \eqref{eq:tfi-ass} quantifies over; therefore, this is implied by 
	\eqref{eq:class-ind-suc}.
	Similarly, if $\islim(a)$, the assumption of \eqref{eq:tfi-ass} gives $\forall i. P(f_i)$, thus \eqref{eq:class-ind-lim} gives $P(a)$.
\end{proof}
\addtocounter{theorem}{-1}
\endgroup

\begin{lemma}\label{thm:CFI-to-classification}
If $A$ satisfies the principle of classifiability induction then it also has classification.
\end{lemma}
\begin{proof}
By \cref{lem:only-one-out-of-three}, we can use the type
\begin{equation}
\iszero(a) \dissum \isstrongsuc(a) \dissum \islimit(a)
\end{equation} 
as the motive of a straightforward classifiability induction. 
\end{proof}


\section{Proofs for Cantor Normal Forms} \label{app:cnf}

We firstly prove the theorems about properties of $\cnf$ and its relations.

\begingroup
\def\thetheorem{\ref{thm:general-notions}}
\begin{theorem}[for Cantor normal forms]
	$\cnf$ is a set and the relations $<$ and $\leq$ are valued in propositions.
	In addition, $<$ and $\leq$ are transitive, $<$ is irreflexive, and $\leq$ is reflexive and antisymmetric.
	$<$ is trichotomous and $\leq$ connex.
\end{theorem}
\begin{proof}
	Most properties are proved by induction on the arguments. We prove the trichotomy property as an example. It is trivial when either argument is zero. Given $\tom a b$ and $\tom c d$, by induction hypothesis we have $a < c \uplus a = c \uplus c < a$ correspondingly. For the first and last cases, we have $\tom a b < \tom c d$ and $\tom c d < \tom a b$. For the middle case, the induction hypothesis on $b$ and $d$ gives the desired result.
\end{proof}
\addtocounter{theorem}{-1}
\endgroup

\begingroup
\def\thetheorem{\ref{thm:<-extensional}}
\begin{theorem}[for Cantor normal forms]
	The relations $<$ and $\leq$ on $\cnf$ are extensional.
\end{theorem}
\begin{proof}
	Extensionality is trivial for an irreflexive and trichotomous relation, which covers $<$, and is also trivial for a reflexive and antisymmetric relation, which covers $\leq$.
\end{proof}
\addtocounter{theorem}{-1}
\endgroup

\begingroup
\def\thetheorem{\ref{thm:<-wellfounded}}
\begin{theorem}[for Cantor normal forms]
	The relation $<$ on $\cnf$ is wellfounded.
\end{theorem}
\begin{proof}
	See the proof of \cite[Theorem~5.1]{NFXG:three:ord}.
\end{proof}
\addtocounter{theorem}{-1}
\endgroup

\begingroup
\def\thetheorem{\ref{thm:all-satisfy-assumptions}}
\begin{theorem}[for Cantor normal forms]
	$\cnf$ and its relations satisfy:
\begin{enumerate}[({A}1)]
	\item \label{item:assumption1Cnf}
	$<$ is transitive and irreflexive;
	\item \label{item:assumption2Cnf}
	$\leq$ is reflexive, transitive, and antisymmetric;
	\item \label{item:assumption3Cnf}
	we have $(<) \subseteq (\leq)$ and $(< \circ \leq) \subseteq (<)$.
\end{enumerate}	
\end{theorem}
\begin{proof}
	By induction on arguments as in the proof of \cref{thm:general-notions}.
\end{proof}
\addtocounter{theorem}{-1}
\endgroup




\begingroup
\def\thetheorem{\ref{thm:all-zero}}
\begin{theorem}[for Cantor normal forms]
	$\cnf$ has a zero.
\end{theorem}
\begin{proof}
	The element $\tz : \cnf$ is a zero.
\end{proof}
\addtocounter{theorem}{-1}
\endgroup

\begingroup
\def\thetheorem{\ref{thm:strong-succs}}
\begin{theorem}[for Cantor normal forms]
	$\cnf$ has strong successors, and it is both $<$- and $\leq$-monotone.
\end{theorem}
\begin{proof}
	It is trivial to show that $a + \tone$ is a strong successor of $a$. The monotonicity properties of $\_+\tone$ is proved by induction on the argument.
\end{proof}
\addtocounter{theorem}{-1}
\endgroup

\begingroup
\def\thetheorem{\ref{thm:sups-lims}}
\begin{theorem}[for Cantor normal forms]
  $\cnf$ does not have suprema or limits.
  Assuming the law of exclude middle, $\cnf$ has suprema (and thus limits) of arbitrary \emph{bounded} sequences. If $\cnf$ has limits of bounded increasing sequences, then the weak limited principle of omniscience (WLPO) is derivable.
\end{theorem}
\begin{proof}
To show that $\cnf$ does not have suprema or limits, we use the sequence $\omega\uparrow\uparrow$ from \cref{thm:cnf-below-eps0} as a counterexample. Recall $\omega \uparrow\uparrow 0 \defeq \omega$ and $\omega \uparrow\uparrow (k+1) \defeq \omega^{\omega \uparrow\uparrow k}$. If it has a limit, say $x : \cnf$, then any CNF $a$ is strictly smaller than $x$, including $x$ itself. But this contradicts to irreflexivity.

Assume $\LEM$. Given a bounded sequence $f$ with bound $b$, consider the subset $P : \cnf \to \Prop$ of all Cantor normal forms that are upper bounds of $f$.
This subset contains at least $b$ and is thus non-empty.
Since $<$ on $\cnf$ is extensional and wellfounded, \cite[Thm~10.4.3]{hott-book} implies that $P$ has a least element, and this least upper bound is a supremum of $f$.

The last part uses arithmetic for Cantor normal forms.
If $\cnf$ has limits of bounded increasing sequences, assume $s : \N \to \Bool$ is a given sequence.
Define $f : \N \to \cnf$ by $f_0 \defeq 0$ and 
\begin{equation}
f_{n+1} \defeq 
\begin{cases}
f_n + 1 & \text{if } s_n = \mathsf{ff} \\
\mathsf{max}(f_n, \omega) + 1 &\text{if } s_n = \mathsf{tt}.
\end{cases}
\end{equation}
	The increasing sequence is bounded by $\omega + \omega$. We can check whether its limit is $\omega$, in which case the original sequence $s$ is constantly $\mathsf{ff}$, otherwise it is not.
\end{proof}
\addtocounter{theorem}{-1}
\endgroup

Even though we cannot compute limits of arbitrary sequences in $\cnf$, if a $x$ is not zero or a successor, then we can always find a \emph{fundamental sequence} whose limit is $x$ --- this is enough to establish the following:

\begin{lemma} \label{lem:cnf:limit}
If a CNF is neither zero nor a successor, then it is a limit.
\end{lemma}
\begin{proof}
If a CNF $x$ is neither zero nor a successor, then $x = \tom a \tz$ with $a>0$ or $x = \tom a b$ where $b > 0$ is not a successor. There are three possible cases, for each of which we construct an increasing sequence $s : \N \to \cnf$ whose limit is $x$:
\begin{enumerate}[(i)]
\item If $x = \tom a \tz$ and $a=c+1$, we define $s_i \defeq (\tom c 0) \cdot \eta i$ where $\eta : \N \to \cnf$ embeds natural numbers to CNFs. 
\item If $x = \tom a \tz$ and $a$ is not a successor, the induction hypothesis on $a$ gives a sequence $r$, and then we define $s_i \defeq \tom {r_i} \tz$. 
\item If $x = \tom a b$ and $b>0$ is not a successor, the induction hypothesis on $b$ gives a sequence $r$, and then we define $s_i \defeq \tom a {r_i}$. 
  \qedhere
\end{enumerate}
\end{proof}

\begingroup
\def\thetheorem{\ref{thm:classification}}
\begin{theorem}[for Cantor normal forms]
	$\cnf$ has classification.
\end{theorem}
\begin{proof}
	Since $\cnf$ has decidable equality, being zero and being a successor are both decidable. Then \cref{lem:cnf:limit} gives the desired result.
\end{proof}
\addtocounter{theorem}{-1}
\endgroup

\begingroup
\def\thetheorem{\ref{thm:classifyability-induction}}
\begin{theorem}[for Cantor normal forms]
	$\cnf$ satisfies classifiability induction.
\end{theorem}
\begin{proof}
	By \cref{thm:<-wellfounded,thm:classification,thm:CFI}. 
\end{proof}
\addtocounter{theorem}{-1}
\endgroup

To show that the operations $(+)$ and $(\cdot)$ correctly implement addition and multiplication in the sense of \cref{def:have-addition,def:have-multiplication}, we construct their inverse operations \emph{subtraction} and \emph{division}:
\begin{lemma}[Subtraction and division of CNFs] \label{lem:sub-div}
For all CNFs $a,b$,
\begin{enumerate}
\item \label{item:one}
if $a \leq b$, then there is a CNF $c$ such that $a+c=d$;
\item \label{item:two}
if $b > \tz$, then there are CNFs $c$ and $d$ such that $a=b \cdot c + d$ and $d<b$. \qedhere
\end{enumerate}
\end{lemma}
\begin{proof}
For \ref{item:one} we define subtraction as follows:
\begin{alignat*}{3}
    \tz - b & \defeq \tz \\
    a - \tz & \defeq a \\
    (\tom{a}{c}) - (\tom{b}{d}) & \defeq
    \begin{cases}
    \tz & \text{if $a<b$} \\
    c-d & \text{if $a=b$} \\
    \tom{a}{c} & \text{if $a>b$}.
    \end{cases}
\end{alignat*}
The proof of \ref{item:two} consists of the following cases:
\begin{itemize}
\item If $a<b$, then we take $c \defeq \tz$ and $d \defeq a$.
\item If $a=b$, then we take $c \defeq \tone$ and $d \defeq \tz$.
\item If $a>b$, then there two possibilities:
\begin{enumerate}[(i)]
\item $a=\tom u u'$ and $b=\tom v v'$ with $u>v$. By the induction hypothesis on $u'$ and $b$, we have $c'$ and $d$ such that $u'=b \cdot c' + d$ and $d<b$. We take $c \defeq \tom{(u-v)}{c'}$ and then have $a = \tom u u' = b \cdot \upomega^{(a-b)} + b \cdot c' + d = b \cdot c + d$.
\item $a=\tom u u'$ and $b=\tom u v'$ with $u'>v'$. By the induction hypothesis on $u'-v'$ and $b$, we have $c'$ and $d$ such that $u'-v' = b \cdot c' + d$ and $d<b$. We take $c \defeq c' + \tone$ and then have $a = \tom u u' = \upomega^u + v' + (u' - v') = b + b \cdot c' + d = b \cdot c + d$.
\end{enumerate}
\end{itemize}
The above defines the (Euclidean) division of CNFs.
\end{proof}

\begingroup
\def\thetheorem{\ref{thm:all-arithmetic-operations}}
\begin{theorem}[for Cantor normal forms]
	$\cnf$ has addition, multiplication, and exponentiation with base $\upomega$, and all operations are unique.
\end{theorem}
\begin{proof}
As an example, we verify that $(+)$ implements addition, i.e., it satisfies
\begin{enumerate}
\item \label{item:zer} $\iszero(a) \to c + a = c$,
\item \label{item:suc} $\issucof a b \to \issucof {d} {(c+b)} \to c + a = d$, and
\item \label{item:lim} $\islimof{a}{f} \to \issupof b {(\lambda i. c + f_i)} \to c + a = b$.
\end{enumerate}
For~\ref{item:zer}, it suffices to show that $\tz$ is the additive identity, which is easy. For~\ref{item:suc}, it suffices to show $c + (b+1) = (c+b)+1$ which holds by associativity of $(+)$. For~\ref{item:lim}, it suffices to show that $c+a$ is the limit of $\lambda i.c+f_i$. Firstly, we have $c+f_i \leq c + a$ for all $i$ because $a$ is the limit of $f$ and $(+)$ is monotone on the right argument. It remains to show that $c+a$ is at least as large as any $x$ with $c+f_i \leq x$ for all $i$. Given such $x$, we have $f_i \leq x-c$ for all $i$. Because $a$ is the limit of $f$, we have $a \leq x-c$. Therefore, we have $c+a \leq c + (x-c) = x$ where the equation holds due to Lemma~\ref{lem:sub-div}.\ref{item:one}. 
\end{proof}
\addtocounter{theorem}{-1}
\endgroup

\section{Proofs about Brouwer Trees}
\label{app:brouwer}

\begin{lemma} \label{lem:distincton-of-constructors}
	The constructors of $\brouwer$ are distinguishable, i.e.\ one can construct proofs of $\bzero \neq \bsuc \, x$, $\bzero \neq \blimit \, f$, and $\bsuc \, x \neq \blimit \, f$.
\end{lemma}
\begin{proof}
Point constructors are not always distinct in the presence of path constructors.
Nevertheless, this is fairly simple in our case, as the path constructor $\bbisim$ only equates limits, and the standard strategy of simply defining distinguishing families such as $\isZero : \brouwer \to \Prop$ works.
Setting
\begin{alignat*}{4}
&\isZero &\;& \bzero &\, & \defeq &\quad &  \Unit \\
&\isZero && (\bsuc \, x) && \defeq &&  \Empty \\
&\isZero && (\blimit \, f) && \defeq &&  \Empty 
\end{alignat*}
means the proof obligations for the path constructors ($\bbisim$ and the truncation constructor) are trivial. Now if $\bzero = \bsuc \, x$, since $\isZero \; \bzero$ is inhabited, $\isZero\, (\bsuc \, x) \equiv \Empty$ must be as well --- a contradiction, which shows $\neg (\bzero = \bsuc \, x)$.

In the same way, one shows the other parts. For the details, we refer to our formalisation.
\end{proof}

\subsection*{Characterisation of $\leq$}

In order to prove some of the non-trivial properties of $\brouwer$, it is necessary to characterise the relation $\leq$.
We use a strategy corresponding to the \emph{encode-decode method} \cite{licataShulman_circle} and define the type family
\begin{equation*}
\Code : \brouwer \to \brouwer \to \Prop
\end{equation*}
mutually with proofs of its \emph{correctness} properties
\begin{alignat*}{5}
& \toCode : &\; & x \leq y \to \Code \, x \, y \\
& \fromCode : && \Code \, x \, y \to x \leq y,
\end{alignat*}
for every $x, y : \brouwer$,
with the goal of providing a concrete description of $x \leq y$.
On the point constructors, the definition works as follows:
\newcommand{\noparen}{\phantom{(}}
\begin{alignat}{7}
&\Code &\;\;& \noparen\bzero &\;&  \noparen \blank &\; &\defeq &\quad &  \Unit \nonumber \\
&\Code && (\bsuc \, x) &&  \noparen\bzero &&\defeq &&  \Empty \nonumber \\
&\Code && (\bsuc \, x) &&  (\bsuc \, y) &&\defeq &&  \Code \; x \; y \label{eq:code-suc-suc}\\
&\Code && (\bsuc \, x) &&  (\blimit \, f) &&\defeq &&  \exists n. \Code \, (\bsuc \, x) \, (f \, n) \label{eq:code-suc-limit} \\
&\Code && (\blimit \, f) &&  \noparen\bzero &&\defeq &&  \Empty \nonumber \\
&\Code && (\blimit \, f) &&  (\bsuc \, y) &&\defeq &&  \forall k. \Code \, (f \, k) \, (\bsuc \, y) \label{eq:code-limit-suc}\\
&\Code && (\blimit \, f) &&  (\blimit \, g) &&\defeq &&
\forall k. \exists n.  \Code \, (f \, k) \, (g \, n)
\label{eq:code-limit-limit}
\end{alignat}

As explained in the paper, the path constructor $\bbisim$ proves to be challenging.
We refer to our Agda formalisation for the details.

\begin{lemma} \label{lem:brouwer-code}
	For $x,y : \brouwer$, the type $\Code \, x \, y$ is equivalent to $x \leq y$.
\end{lemma}
\begin{proof}
	Both types are propositions, and maps in both directions are given by $\toCode$ and $\fromCode$.
\end{proof}

$\Code$ allows us to easily derive various useful auxiliary lemmas, for example the following four:

\begin{lemma} \label{lem:suc-mono-inv}
	For $x,y : \brouwer$, we have $x \leq y \leftrightarrow \bsuc \, x \leq \bsuc \, y$.
\end{lemma}
\begin{proof}
	This is a direct consequence of \eqref{eq:code-suc-suc} and the correctness of $\Code$.
\end{proof}

\begin{lemma} \label{lem:x-under-limit-means-x-under-element}
	Let $f : \N \toincr{<} \brouwer$ be an increasing sequence and $x : \brouwer$ a Brouwer tree such that $x < \blimit \, f$.
	Then, there exists an $n : \N$ such that $x < f \, n$.
\end{lemma}
\begin{proof}
	By definition, $x < \blimit \, f$ means $\bsuc \, x \leq \blimit \, f$.
	Using $\toCode$
	together with the case
	\eqref{eq:code-suc-limit},
	there exists an $n$ such that $\Code \, (\bsuc \, x) \, (f \, n)$.
	Using $\fromCode$, we get the result.
\end{proof}

\begin{lemma}\label{lem:lim-under-lim-simulation}
	If $f$, $g$ are increasing sequences such that $\blimit \, f \leq \blimit \, g$, then $f$ is simulated by $g$.
\end{lemma}
\begin{proof}
	For a given $k$, \eqref{eq:code-limit-limit} tells us that there exists an $n$ such that, after using $\fromCode$, we have $f \, k \leq g \, n$.
\end{proof}

\begin{lemma} \label{lem:lim-under-suc}
	If $f$ is an increasing sequence and $x$ a Brouwer tree such that $\blimit \, f \leq \bsuc \, x$, then $\blimit \, f \leq x$.
\end{lemma}
\begin{proof}
	From \eqref{eq:code-limit-suc}, we have that $f\,k \leq \bsuc \, x$ for every $k$. But since $f$ is increasing, $\bsuc\,(f\,k) \leq f\,(k + 1) \leq \bsuc \, x$ for every $k$, hence by \cref{lem:suc-mono-inv} $f\,k \leq x$ for every $k$, and the result follows using the constructor $\llimiting$.
\end{proof}

\begingroup
\def\thetheorem{\ref{thm:general-notions}}
\begin{theorem}[for Brouwer trees]
	$\brouwer$ is a set,  and $<$ as well as $\leq$ are valued in propositions.
	Both $<$ and $\leq$ are transitive, $<$ is irreflexive, and $\leq$ is reflexive and antisymmetric.
\end{theorem}
\begin{proof}
	$\brouwer$ and $\leq$ are constructed to be a set and valued to be in propositions, respectively. $<$ is defined as an instance of $\leq$ and thus valued in propositions.
	Transitivity for both relations is given by $\ltrans$.
	
	Reflexivity of $\leq$ holds by easy induction.
	Irreflexivity follows directly from wellfoundedness, which for $<$ is given in \cref{thm:<-wellfounded} below.

	Antisymmetry of $\leq$ is shown with the characterisation of $\leq$ given by \cref{lem:brouwer-code}; we sketch how:
	Let $x, y$ with $x \leq y$ and $y \leq x$ be given.
	We do nested induction, i.e.\ induction on both $x$ and $y$. Since we are proving a proposition, we can disregard the cases for path constructors, giving us $9$ cases in total, many of which are duplicates.
	We discuss the two most interesting cases:
	\begin{itemize}
		\item $x \equiv \blimit \, f$ and $y \equiv \bsuc \, y'$: In that case, \cref{lem:lim-under-suc} and the assumed inequalities show $\bsuc \, y' \leq \blimit \,f \leq y'$ and thus $y' < y'$, contradicting the wellfoundedness of $<$ (\cref{thm:<-wellfounded}).
		\item $x \equiv \blimit \, f$ and $y \equiv \blimit \, g$: By \cref{lem:lim-under-lim-simulation}, $f$ and $g$ simulate each other. By the constructor $\bbisim$, $x = y$. \qedhere
	\end{itemize}
\end{proof}
\addtocounter{theorem}{-1}
\endgroup

\begingroup
\def\thetheorem{\ref{thm:<-extensional}}
\begin{theorem}[for Brouwer trees]
	The relations $<$ and $\leq$ on $\brouwer$ are
	extensional.
\end{theorem}
\begin{proof}
	Extensionality is trivial for a reflexive and antisymmetric relation, which covers $\leq$.
	For $<$,
	let $x$ and $y$ be two elements of $\brouwer$ with the same set of smaller elements.
	As in the above proof, we can disregard the path constructors and get $9$ cases. 
	If $x$ and $y$ are built of different constructors, it is easy to derive a contradiction. For example, in the case $x \equiv \blimit \, f$ and $y \equiv \bsuc \, y'$, we have $y' < y$ and thus $y' < \blimit \, f$.
	By \cref{lem:x-under-limit-means-x-under-element},
	there exists an $n$ such that $y' < f \, n$, which in turn implies $y < f \, (n+1)$ and thus $y < \blimit \, f$. By the assumed set of smaller elements, that means we have $y < y$, contradicting irreflexivity of $<$.
	
	The other interesting case, $x \equiv \blimit \, f$ and $y \equiv \blimit \, g$, is easy. For all $k : \N$, we have $f \, k < f \, (k+1) \leq \blimit \, f$ and thus $f \, k < \blimit \, g$; by the constructor $\llimiting$, this implies $\blimit \, f \leq \blimit \, g$.
	By the symmetric argument and by antisymmetry of $\leq$, it follows that $\blimit \, f = \blimit \, g$.
\end{proof}
\addtocounter{theorem}{-1}
\endgroup

\begingroup
\def\thetheorem{\ref{thm:<-wellfounded}}
\begin{theorem}[for Brouwer trees]
	The relation $<$ on $\brouwer$ is wellfounded.
\end{theorem}
\begin{proof}
	We need to prove that every $y : \brouwer$ is accessible.
	When doing induction on $y$, the cases of path constructors are automatic as we are proving a proposition.
	From the remaining constructors, we only show the hardest case, which is that $y \equiv \blimit \, f$.
	We have to prove that any given $x < \blimit \, f$ is accessible.
	By \cref{lem:x-under-limit-means-x-under-element}, there exists an $n$ such that $x < f\, n$, and the latter is accessible by the induction hypothesis.
\end{proof}	
\addtocounter{theorem}{-1}
\endgroup

\begingroup
\def\thetheorem{\ref{thm:all-satisfy-assumptions}}
\begin{theorem}[for Brouwer trees]
	$\brouwer$ and its relations satisfy:
	\begin{enumerate}[({A}1)]
		\item \label{item:assumption1C}
		$<$ is transitive and irreflexive;
		\item \label{item:assumption2C}
		$\leq$ is reflexive, transitive, and antisymmetric;
		\item \label{item:assumption3C}
		we have $(<) \subseteq (\leq)$ and $(< \circ \leq) \subseteq (<)$.
	\end{enumerate}	
\end{theorem}
\begin{proof}
	The first two properties have been shown above. $(<) \subseteq (\leq)$ follows easily from \cref{lem:brouwer-code} and $(< \circ \leq) \subseteq (<)$ by transitivity.
\end{proof}
\addtocounter{theorem}{-1}
\endgroup

\begingroup
\def\thetheorem{\ref{thm:all-zero}}
\begin{theorem}[for Brouwer trees]
	$\brouwer$ has a zero.
\end{theorem}
\begin{proof}
	The zero is given directly by the constructor $\bzero$, with $\lzero$ witnessing its property.
\end{proof}
\addtocounter{theorem}{-1}
\endgroup

\begingroup
\def\thetheorem{\ref{thm:strong-succs}}
	\begin{theorem}[for Brouwer trees]
	$\brouwer$ has strong successors,
	and
	the successor function is both $<$- and $\leq$-monotone.
\end{theorem}
\begin{proof}
	The function is given by the constructor $\bsuc$.
	It calculates successors by \cref{thm:calc-succ-characterisation} and by definition.
	To verify that it is a strong successor we need to show that $x < \bsuc \, b$ implies $x \leq b$.
	But $x < \bsuc \, b$ is defined to mean $\bsuc \, x \leq \bsuc \, b$ which, by \cref{lem:suc-mono-inv}, is indeed equivalent to $x \leq b$.	
\end{proof}
\addtocounter{theorem}{-1}
\endgroup

\begingroup
\def\thetheorem{\ref{thm:sups-lims}}
	\begin{theorem}[for Brouwer trees]
		$\brouwer$ has limits of increasing sequences.
	\end{theorem}
	\begin{proof}
		Directly given by the constructor $\blimit$ in conjunction with the constructors $\lcocone$ and $\llimiting$.
	\end{proof}
	\addtocounter{theorem}{-1}
	\endgroup

\begingroup
\def\thetheorem{\ref{thm:classification}}
\begin{theorem}[for Brouwer trees]
	$\brouwer$ has classification.
\end{theorem}
	\begin{proof}
	By \cref{thm:CFI-to-classification} and \cref{thm:classifyability-induction} below.
\end{proof}
\addtocounter{theorem}{-1}
\endgroup

\begingroup
\def\thetheorem{\ref{thm:classifyability-induction}}
\begin{theorem}[for Brouwer trees]
	$\brouwer$ satisfies classifiability induction.
\end{theorem}
\begin{proof}
	Since $\brouwer$ has zero, successor, and limits of increasing sequences, all given by the respective constructor, the special case of induction for $\brouwer$ where the goal is a proposition is exactly classifiability induction.
\end{proof}
\addtocounter{theorem}{-1}
\endgroup

\begingroup
\def\thetheorem{\ref{thm:all-arithmetic-operations}}
\begin{theorem}[for Brouwer trees]
	$\brouwer$ has addition, multiplication and exponentiation with every base, and they are all unique.
\end{theorem}
\begin{proof}
The standard arithmetic operations on Brouwer trees can be implemented with the usual strategy well-known in the functional programming community, i.e.\ by recursion on the second argument.
However, there are several additional difficulties which stem from the fact that our Brouwer trees enforce correctness.

Let us start with addition.
The obvious definition is
\begin{alignat*}{10}
&x &\; & + &\; & \bzero  &\quad & \defeq  &\quad & x \\
&x && + && \bsuc \, y  && \defeq  && \bsuc \, (x + y) \\
&x && + && \blimit \, f  && \defeq  && \blimit \, (\lambda k. x + f \, k).
\end{alignat*}
For this to work, we need to prove, mutually with the above definition, that the sequence $\lambda k. x + f\, k$ in the last line is still increasing, which follows from mutually proving that $+$ is monotone in the second argument, both with respect to $\leq$ and $<$.
We also need to show that bisimilar sequences $f$ and $g$ lead to bisimilar sequences $x + f\, k$ and $x + g\, k$.

The same difficulties occur for multiplication $(\cdot)$, where they are more serious:
If $f$ is increasing (with respect to $<$), then
$\lambda k. x + f \, k$ is not necessarily increasing, as $x$ could be $\bzero$.
What saves us is that it is \emph{decidable} whether $x$ is $0$ (cf.~\cref{lem:distincton-of-constructors}); and if it is, the correct definition is $x \cdot \blimit \, f \defeq \bzero$.
If $x$ is not $0$, then it is at least $1$ (another simple lemma for $\brouwer$), and the sequence \emph{is} increasing.
With the help of several lemmas that are all stated and proven mutually with the actual definition of $(\cdot)$, 
the mentioned decidability is the core ingredient which allows us to complete the construction.
Exponentation $x^y$ comes with similar caveats as multiplication, but works with the same strategy.

Our Agda formalisation contains proofs that addition, multiplication, and exponentation have the properties introduced in \cref{subsec:abstract-arithmetic}, and that they are all unique. 
\end{proof}
\addtocounter{theorem}{-1}
\endgroup

The formalisation contains several other results which are omitted in the paper. Define $\iota : \N \to \brouwer$ as the embedding of natural numbers into Brouwer trees, i.e.\ $\iota(0) \defeq \bzero$ and $\iota(n+1) \defeq \bsuc\,(\iota(n))$, and let $\omega \defeq \blimit\,\iota$. We then have:

\begin{lemma}
	\label{thm:additive-principal}
	Brouwer trees of the form $\omega^x$ are additive principal:
	\begin{enumerate}[(a)]
		\item if $a < \omega^x$ then $a + \omega^x = \omega^x$. \label{item:add-prin-1}
		\item if $a < \omega^x$ and $b < \omega^x$ then $a + b < \omega^x$. \label{item:add-prin-2}
	\end{enumerate}
\end{lemma}
\begin{proof}
	For \eqref{item:add-prin-1}, it is enough to show $a + \omega^x \leq \omega^x$ by monotonicity of $+$. We prove this by induction on $a$, making crucial use of 
	\cref{lem:x-under-limit-means-x-under-element}.
	Statement \eqref{item:add-prin-2} is an easy corollary of \ref{item:add-prin-1}) and strict monotonicity of $+$ in the second argument.
\end{proof}

\begin{lemma}
	\label{thm:omega-preserves-finite}
	If $x > \bzero$ and $n : \N$, then $\iota(n) \cdot \omega^x = \omega^x$.
\end{lemma}
\begin{proof}
  By induction on $x$. The base case $x = \bzero$ is impossible by
  assumption. If $x = \bsuc\, y$, we can decide if $y = \bzero$ or
  not. If it is, i.e.\ if $x = \bsuc\,\bzero$, we prove
  $\iota(n) \cdot \omega= \omega$ by induction on $n$, and otherwise
  the goal follows straightforwardly from the induction hypothesis.
  In the limit case $x = \blimit\,f$, we use that
  $\blimit\,f = \blimit\,(\lambda k . f\,(k+1))$, and that $f\,(k+1) > \bzero$ for every increasing sequence $f$, in order for the induction hypothesis to apply.
\end{proof}


\section{Proofs about Extensional Wellfounded Orders}
\label{app:orders}

\begingroup
\def\thetheorem{\ref{thm:general-notions}}
\begin{theorem}[for extensional wellfounded orders]
	$\bookord$ is a set and the relations $<$ and $\leq$ are valued in propositions.
	In addition, $<$ and $\leq$ are transitive, $<$ is irreflexive, and $\leq$ is reflexive and antisymmetric.
	$<$ is trichotomous iff $\leq$ is connex iff $\LEM$ holds.
\end{theorem}
\begin{proof}
	All statements follow either from \cite[Thm~10.3.20]{hott-book} or are included in
	\cref{thm:ord-everything-undecidable} below.
\end{proof}
\addtocounter{theorem}{-1}
\endgroup

\begingroup
\def\thetheorem{\ref{thm:<-extensional}}
\begin{theorem}[for extensional wellfounded orders]
	The relations $<$ and $\leq$ on $\bookord$ are extensional.
\end{theorem}
\begin{proof}
	Extensionality is trivial for a reflexive and antisymmetric relation, which covers $\leq$.
	Extensionality of $<$ is shown in \cite[Thm~10.3.20]{hott-book}.
\end{proof}
\addtocounter{theorem}{-1}
\endgroup

\begingroup
\def\thetheorem{\ref{thm:<-wellfounded}}
\begin{theorem}[for extensional wellfounded orders]
	The relation $<$ on $\bookord$ is wellfounded.
\end{theorem}
\begin{proof}
	See \cite[Thm~10.3.20]{hott-book}.
\end{proof}
\addtocounter{theorem}{-1}
\endgroup

\begingroup
\def\thetheorem{\ref{thm:all-satisfy-assumptions}}
\begin{theorem}[for extensional wellfounded orders]
	$\bookord$ and its relations satisfy:
	\begin{enumerate}[({A}1)]
		\item \label{item:assumption1ord}
		$<$ is transitive and irreflexive;
		\item \label{item:assumption2ord}
		$\leq$ is reflexive, transitive, and antisymmetric;
		\item \label{item:assumption3ord}
		we have $(<) \subseteq (\leq)$ and $(< \circ \leq) \subseteq (<)$.
	\end{enumerate}	
\end{theorem}
\begin{proof}
	The first two properties follow from $\bookord$ being an ordinal \cite[Thm~10.3.20]{hott-book}.
	The first part of \ref{item:assumption3ord} is projection, the second holds by composition of functions.
\end{proof}
\addtocounter{theorem}{-1}
\endgroup

The property symmetric to $(< \circ \leq) \subseteq (<)$ fails for $\bookord$:

\begin{lemma}
	Let $A$, $B$, $C : \bookord$. The statement that $A \leq B$ and $B < C$ implies $A < C$ holds if and only if excluded middle $\LEM$ does.
\end{lemma}
\begin{proof}
	If $A \simeq B_{\slash b}$ and $g : B \leq C$ then $A \simeq
	C_{\slash g(b)}$.  Assuming $\LEM$ and $f : A \leq B$, $g : B <
	C$, there is a minimal $c : C$ not in the image of $g \circ
	f$ by \cite[Theorem 10.4.3]{hott-book}, and $A \simeq C_{\slash
		c}$.  Conversely, let $P$ be a proposition; it is an ordinal with the empty
	order. We have $\Unit \leq \Unit \dissum P$ and $\Unit \dissum P < \Unit \dissum P \dissum
	\Unit$, so by assumption $\Unit < \Unit \dissum P \dissum
	\Unit$. Now observe which component of the sum the bound is from: this shows if $P$ holds or not.
\end{proof}

\begingroup
\def\thetheorem{\ref{thm:all-zero}}
\begin{theorem}[for extensional wellfounded orders]
	$\bookord$ has a zero.
\end{theorem}
\begin{proof}
	Zero is given as the empty type $\Empty$.
\end{proof}
\addtocounter{theorem}{-1}
\endgroup

\begingroup
\def\thetheorem{\ref{thm:strong-succs}}
\begin{theorem}[for extensional wellfounded orders]
	$\bookord$ has strong successors, with the successor of $A$ given as $\osuc{A}$, but monotonicity with respect to either $<$ or $\leq$ is equivalent to $\LEM$.
\end{theorem}
\begin{proof}
	The definition of a bounded simulation implies $(X < Y) \leftrightarrow (X \dissum \Unit \leq Y)$, 
	making
	\cref{thm:calc-succ-characterisation}
	applicable.
	For the rest of the statement, see \cref{thm:ord-everything-undecidable} below.
\end{proof}
\addtocounter{theorem}{-1}
\endgroup

\begingroup
\def\thetheorem{\ref{thm:sups-lims}}
\begin{theorem}[for extensional wellfounded orders]
If $F : \N \to \bookord$ is an increasing sequence of ordinals, then its limit $\osup{F}$ is the quotient $(\Sigma n : \N.F n)/\sim$, where $(n, y) \sim (n', y')$ if and only if $(F n)_{\slash y} \simeq (F n')_{\slash y'}$, with $[(n, y)] \prec [(n',y')]$ if $(F n)_{\slash y} < (F n')_{\slash y'}$.
If $F$ is only weakly increasing (i.e.\ increasing with respect to $\leq$), the same formula defines the supremum of $F$.
\end{theorem}
\begin{proof}
The construction of limits/suprema is from \cite[Lem~10.3.22]{hott-book}, where it is shown that the canonical function $l_n : F n \to \osup F$ is a simulation.
A family $g_n : F n \leq X$ readily gives a function $(\Sigma n:\N. Fn) \to X$, which we extend to the quotient. Given $(n,y) \sim (n', y')$, let WLOG $n \leq n'$; hence there is $f : F n \leq F n'$. The condition $(F n)_{\slash y} \simeq (F n')_{\slash y'}$ implies $f \, y = y'$.
As simulations are unique, we have $l_{n} = l_{n'} \circ f$, which means that $(n,y)$ and $(n',y')$ are indeed mapped to equal elements in $\osup F$, giving us a map $g : \osup F \to X$.
As every $g_n$ is a simulation, so is $g$.
\end{proof}
\addtocounter{theorem}{-1}
\endgroup

Note that the above proof explicitly uses that $F$ is (at least weakly) increasing; we do not think that the more general construction in \cite[Lem~10.3.22]{hott-book}
can be shown constructively to give suprema in our sense.

\begingroup
\def\thetheorem{\ref{thm:classification}}
\begin{theorem}[for extensional wellfounded orders]
	If $\bookord$ has classification, the law of excluded middle is derivable.
\end{theorem}
\begin{proof}
	See \cref{thm:ord-everything-undecidable}.
\end{proof}
\addtocounter{theorem}{-1}
\endgroup

\begingroup
\def\thetheorem{\ref{thm:classifyability-induction}}
\begin{theorem}[for extensional wellfounded orders]
	If $\bookord$ satisfies classifiability induction, we can derive the law of excluded middle.
\end{theorem}
\begin{proof}
	See \cref{thm:ord-everything-undecidable}.
\end{proof}
\addtocounter{theorem}{-1}
\endgroup

\begingroup
\def\thetheorem{\ref{thm:all-arithmetic-operations}}
\begin{theorem}[for extensional wellfounded orders]
	$\bookord$ has addition given by $A + B = A \dissum B$, and multiplication given by $A \cdot B = A \times B$, with the order reverse lexicographic, i.e.\ $(x, y) \prec (x' , y')$ is defined to be $y \prec_B y' \dissum (y = y' \times x \prec_A x')$.
\end{theorem}
\begin{proof}
	The key observation is that a sequence of simulations $F0 \leq F1 \leq F2 \leq \ldots$ is preserved by adding or multiplying a constant \emph{on the left}, i.e.\ we have $C \cdot F0 \leq C \cdot F1 \leq C \cdot F2$ (but note that adding a constant on the right fails, see \cref{thm:ord-everything-undecidable} below).
	This allows us to use the explicit representation of suprema from \cref{thm:sups-lims}.
\end{proof}
\addtocounter{theorem}{-1}
\endgroup

Many constructions that we have performed for $\cnf$ and $\brouwer$ are not possible for $\bookord$, at least not constructively:

\begin{theorem} \label{thm:ord-everything-undecidable} Each of the
  following statements on its own implies the law of excluded middle
  ($\LEM$):
  \begin{enumerate}[(a)]
    \item The successor $(\blank \dissum \Unit)$
          is $\leq$-monotone. \label{item:leq-mon}
    \item The successor $(\blank \dissum \Unit)$ is $<$-monotone. \label{item:<-mon}
    \item $<$ is trichotomous, i.e.\ $(X < Y) \dissum (X = Y) \dissum (X > Y)$. \label{item:<-tricho}
    \item $\leq$ is connex, i.e.\ $(X \leq Y) \dissum (X \geq Y)$. \label{item:leq-connex}
    \item $\bookord$ satisfies classifiability induction. \label{item:ord-CFI}
    \item $\bookord$ has classification. \label{item:ord-class}
  \end{enumerate}
  Assuming $\LEM$, the first four statements hold.
\end{theorem}
\begin{proof}
    We first show the chain $\LEM \Rightarrow \eqref{item:leq-mon} \Rightarrow \eqref{item:<-mon} \Rightarrow \LEM$.

    $\LEM \Rightarrow \eqref{item:leq-mon}$:
    Let $f: A \leq B$.
    Using $\LEM$, there is a minimal  $b : B \dissum \Unit$ which is not in the image of $f$~\cite[Theorem 10.4.3]{hott-book}.
    The simulation $A \dissum \Unit \leq B \dissum \Unit$ is given by $f \dissum b$.

    $\eqref{item:leq-mon} \Rightarrow \eqref{item:<-mon}$:
    Assume we have $A < B$. As in \cref{thm:strong-succs}, this is equivalent to $A \dissum \Unit \leq B$.
    Assuming $\simplesuc$ is $\leq$-monotone, we get $A \dissum \Bool \leq B \dissum \Unit$, and applying \cref{thm:calc-succ-characterisation} once more, this is equivalent to $A \dissum \Unit < B \dissum \Unit$.

    $\eqref{item:<-mon} \Rightarrow \LEM$:
    Assume $P$ is a proposition. We have $\Empty < \Unit \dissum P$.
    If $\simplesuc$ is $<$-monotone, then we get $\Empty \dissum \Unit < \Unit \dissum P \dissum \Unit$. Observing if the simulation $f$ sends $\mathsf{inr}(\star)$ to the $P$ summand or not, we decide $P  \dissum  \neg P$.

    Next, we show the chain $\LEM \Rightarrow \eqref{item:<-tricho} \Rightarrow \eqref{item:leq-connex} \Rightarrow \LEM$, where the first implication is given by \cite[Thm 10.4.1]{hott-book}.

    $\eqref{item:<-tricho} \Rightarrow \eqref{item:leq-connex}$: Each of the three cases of \eqref{item:<-tricho} gives us either $X \leq Y$ or $X \geq Y$ or both.

    $\eqref{item:leq-connex} \Rightarrow \LEM$: Given $P : \Prop$, we compare $P \dissum P$ with $\Unit$.
    If $P \dissum P \leq \Unit$ then $\neg P$, while $\Unit \leq P \dissum P$ implies $P$.

    Finally, we show $\eqref{item:ord-CFI} \Rightarrow \eqref{item:ord-class} \Rightarrow \LEM$. The first implication is clear.
    For the second implication, observe that a classifiable proposition is either zero ($\Empty$) or a successor ($X \dissum \Unit$), as a limit is never a proposition. In both cases, we are done.
\end{proof}



\section{Proofs for \texorpdfstring{\cref{sec:interpretations}}{Section 5}}
\label{app:interpretations}

\begingroup
\def\thetheorem{\ref{thm:CtoB-reflects}}
\begin{theorem}
	The function $\CtoB$ preserves and reflects $<$ and $\leq$, i.e., $a < b \iff \CtoB(a) < \CtoB(b)$, and  $a \leq b \iff \CtoB(a) \leq \CtoB(b)$.
\end{theorem}
\begin{proof}
  We show the proof for $<$; each direction of the statement for $\leq$ is a simple consequence.

  ($\Rightarrow$)~By induction on $a<b$. The case when $\tom a b < \tom c d$ because $a < c$ uses \cref{thm:additive-principal}.

  ($\Leftarrow$)~Assume $\CtoB(a)<\CtoB(b)$. If $a \geq b$, then $\CtoB(a) \geq \CtoB(b)$ by ($\Rightarrow$), conflicting the assumption. Hence $a<b$ by the trichotomy of $<$ on $\cnf$.
\end{proof}
\addtocounter{theorem}{-1}
\endgroup

\begingroup
\def\thecorollary{\ref{cor:CtoB-injective}}
\begin{corollary}
	The function $\CtoB$ is injective.
\end{corollary}
\begin{proof}
	$\CtoB(a) = \CtoB(b)$ implies $\CtoB(a) \leq \CtoB(b)$ and thus, by \cref{thm:CtoB-reflects}, $a \leq b$.
	Analogously, one has $b \leq a$. Antisymmetry gives $a = b$.
\end{proof}
\addtocounter{theorem}{-1}
\endgroup

\begingroup
\def\thetheorem{\ref{lem:f-arith}}
\begin{theorem}
	$\CtoB$ commutes with addition, multiplication, and exponentiation with base $\omega$.
\end{theorem}
\begin{proof}
As an example, we show that $\CtoB$ commutes with addition, i.e., $\CtoB(a+b) = \CtoB(a) + \CtoB(b)$ for all $a,b : \cnf$. The proof is carried out by induction on $a,b$. It is trivial when either of them is $\tz$. Assume $a = \tom x u$ and $b = \tom y v$. If $x<y$, then $a+b=b$. We have also $\upomega^x < \upomega^y$, which implies $\omega^{\CtoB(x)} < \omega^{\CtoB(y)}$ by Theorem~\ref{thm:CtoB-reflects}. Then by Lemma~\ref{thm:additive-principal}.\ref{item:add-prin-1} we have $\omega^{\CtoB(x)} + \omega^{\CtoB(y)} = \omega^{\CtoB(y)}$. By the same argument, from the fact $u < \upomega^y$ we derive $\CtoB(u) + \omega^{\CtoB(y)} = \omega^{\CtoB(y)}$. Therefore, both $\CtoB(a+b)$ and $\CtoB(a) + \CtoB(b)$ are equal to $\CtoB(b)$. If $y \leq x$, then $a+b = \tom x {u + b}$ by definition. By induction hypothesis, we have $\CtoB(u + b) = \CtoB(u) + \CtoB(b)$. Therefore, both $\CtoB(a+b)$ and $\CtoB(a) + \CtoB(b)$ are equal to $\omega^{\CtoB(x)} + \CtoB(u) + \CtoB(b)$.

The fact that $\CtoB$ commutes with multiplication can be proved with a similar argument of induction. Moreover, the proof relies on \cref{thm:omega-preserves-finite}.
Preservation of exponentiation with base $\omega$ holds by definition.
%
\end{proof}
\addtocounter{theorem}{-1}
\endgroup

\begingroup
\def\thetheorem{\ref{thm:cnf-below-eps0}}
\begin{theorem}
	For all $a : \cnf$, we have $\CtoB(a) < \blimit\,(\lambda k.\omega \uparrow\uparrow k)$, where $\omega \uparrow\uparrow 0 \defeq \omega$ and $\omega \uparrow\uparrow (k+1) \defeq \omega^{\omega \uparrow\uparrow k}$.
\end{theorem}
\begin{proof}
  By induction on $a$. Using that $\varepsilon_0 = \omega^{\varepsilon_0} = \omega^{\omega^{\varepsilon_0}}$, in the step case we have $\omega^{\CtoB(a)} + \CtoB(b) < \varepsilon_0$ by \cref{thm:additive-principal}, strict monotonicity of $\omega^-$, and the induction hypothesis.
\end{proof}
\addtocounter{theorem}{-1}
\endgroup

\begin{lemma} \label{lem:projection-simulation}
	For $X : \bookord$ with $x : X$, the first projection $\fstproj : X_{\slash x} \to X$ is a simulation.
	If $x, y : X$ and $f : X_{\slash x} \to X_{\slash y}$ is a function, then $f$ is a simulation if and only if $\fstproj \circ f = \fstproj$.
\end{lemma}
\begin{proof}
	Both properties required in the definition of a simulation are obvious in the case of $\fstproj$.
	In the second sentence, if $f$ is a simulation, then the equality follows from the uniqueness of simulations \cite[Thm~10.3.16]{hott-book}.
	If the equality holds then, again, the two properties in the definition of a simulation are clear for $f$.
\end{proof}

\begingroup
\def\thelemma{\ref{lem:BtoO-injective}}
\begin{lemma}
	The function $\BtoO : \brouwer \to \bookord$ is injective, and preserves  $<$ and $\leq$.
\end{lemma}
\begin{proof}
	The first part (injectivity of $\BtoO$) is a special case of the following statement: Given $X : \bookord$, the map $X \to \bookord$, $x \mapsto X_{\slash x}$ is injective. 
	This is remarked just before \cite[before 10.3.19]{hott-book}. We give a detailed proof:
	
	Note that an equality $Y = Z$ in $\bookord$ gives rise to a canonical simulation $X \leq Y$ by path induction.
	Now, assume $x,y : X$ with $X_{\slash x} = X_{\slash y}$.
	We get $f : X_{\slash x} \leq X_{\slash y}$.
	By \cref{lem:projection-simulation}, $f$ maps $(z,p)$ to $(z,q)$, with $q : z < y$; that is, every element below $x$ is also below $y$.
	The symmetric statement follows by the symmetric argument, and injectivity of $x \mapsto X_{\slash x}$ by extensionality.
	
	If $q : x \leq y$ in $\brouwer$, then the map $X_{\slash x} \to X_{\slash y}$, $(z,p) \to (z, p \cdot q)$ is a simulation by \cref{lem:projection-simulation}, thus $\BtoO$ preserves $\leq$.
	This implies that $<$ is preserved as well since $X < Y \leftrightarrow (X \uplus \mathsf 1) \leq Y$.
\end{proof}
\addtocounter{theorem}{-1}
\endgroup

\begingroup
\def\thetheorem{\ref{thm:lem-implies-simulation}}
\begin{theorem}
    Under the assumption of the law of excluded middle, the function $\BtoO : \brouwer \to \bookord$ is a simulation.
\end{theorem}
\begin{proof}
  Given $b < \BtoO(a)$, we need to find a Brouwer tree $a' < a$ such that $\BtoO(a') = b$. Using $\LEM$, we can choose $a'$ to be the minimal Brouwer tree $x$ such that $b \leq \BtoO(x)$.
\end{proof}
\addtocounter{theorem}{-1}
\endgroup

\begingroup
\def\thetheorem{\ref{thm:BtoO-sim-WLPO}}
\begin{theorem}
    If the map $\BtoO : \brouwer \to \bookord$ is a simulation, then WLPO holds.
\end{theorem}
\begin{proof}
    Let a sequence $s$ be given.
    We define the type family $S : \N \to \Prop$ by $S n \defeq (s \, n = \mathsf{tt})$ and regard it as a family of orders.
    We have $\bigsqcup S < \Bool$.
    Observing that $\Bool \equiv \BtoO(\bsuc (\bsuc \, \bzero))$ and using the assumption that $\BtoO$ is a simulation, we get $b < \bsuc (\bsuc \, \bzero)$ such that $\BtoO \, b = \bigsqcup S$.
    It is decidable whether $b$ is $\bzero$ or $\bsuc \, \bzero$, and this determines if $s$ is constantly $\bff$ or not.
\end{proof}
\addtocounter{theorem}{-1}
\endgroup

\end{document}